\documentclass[11pt,longbibliography]{article}

\usepackage[margin=1in]{geometry}
\usepackage{amsmath,amssymb,amsthm,mathtools}
\usepackage{booktabs}
\usepackage{enumitem}
\usepackage{graphicx}
\usepackage[dvipsnames]{xcolor}
\usepackage{algorithm}
\usepackage{algpseudocode}
\usepackage{comment}

\usepackage{authblk}
\usepackage{tikz}
\usetikzlibrary{positioning}
\usepackage{subcaption}

\usepackage[colorlinks=true,citecolor=blue,linkcolor=magenta,urlcolor=magenta]{hyperref}

\usepackage[
   backend=biber,
    style=phys,
    sorting=none,
    biblabel=brackets,
]{biblatex}

\renewbibmacro*{note+pages}{%
   \printfield{note}%
    \setunit{\bibpagespunct}%
    \printfield{pages}%
}

\DeclareFieldFormat[misc]{title}{%
    \printtext[doi/url-link]{\mkbibemph{#1}}%
}

\theoremstyle{plain}
\newtheorem{theorem}{Theorem}
\newtheorem{lemma}{Lemma}

\newtheorem{corollary}{Corollary}
\newtheorem{definition}{Definition}

\DeclareMathOperator{\inv}{inv}
\DeclareMathOperator{\nneg}{neg}
\DeclareMathOperator{\sgn}{sgn}

\newcommand{\id}{\mathbb{I}}

\newcommand{\SO}{\mathbb{SO}}

\newcommand{\SPIN}{\mathbb{SPIN}}

\newcommand{\GC}{\mathcal{G}}

\newcommand{\OC}{\mathcal{O}}

\newcommand{\RC}{\mathcal{R}}

\renewcommand{\vec}[1]{\boldsymbol{#1}}

\begin{document}

\title{Complexity separations for optimal matchgate-Clifford synthesis}

\author[1,2]{Berta Casas}
\affil[1]{Barcelona Supercomputing Center, Plaça Eusebi G\"uell, 1-3, 08034 Barcelona, Spain}
\affil[2]{Information Sciences, Los Alamos National Laboratory, Los Alamos, NM 87545, USA}

\author[3]{Diego Garc\'ia-Mart\'in}
\affil[3]{Department for Quantum Information and Computation at Kepler (QUICK),  Johannes Kepler University, Linz, Austria}

\author[4,5]{Yuxuan Zhang \thanks{\texttt{yuxuan.zhang@epfl.ch}}}
\affil[4]{Institute of Physics, \'Ecole Polytechnique F\'ed\'erale de Lausanne (EPFL), 1015 Lausanne, Switzerland}
\affil[5]{Department of Physics, Princeton University, Princeton, New Jersey 08544, USA}

\date{\vspace{-4ex}}
\maketitle

\begin{abstract}
Clifford and matchgate circuits are canonical families of classically simulable quantum circuits. Their intersection, the matchgate-Clifford group, plays an important role in randomized fermionic protocols and in matchgate synthesis. Its adjoint action is isomorphic to the group of unit-determinant signed permutations of $2n$ Majorana modes, and we study optimal exact synthesis in this group. That is, given a target unitary and a gate set, output an $n$-qubit circuit implementing the target using the fewest operations. We show that the complexity of this problem strongly depends on the gate set. In particular, we study gate sets consisting of Majorana braids with different connectivity graphs. For path and complete graphs, we prove that the problem is classically solvable in $\mathcal{O}\left(n^2\right)$ time, and we provide explicit gate-optimal compilers. In addition, we prove that when the connectivity graph is a tree, the decision version of the optimal synthesis problem becomes NP-complete. Finally,  we benchmark our optimal compiler on chains of up to $n=80$ qubits against those of \texttt{Qiskit} and \texttt{Tket}, obtaining circuits with constant factor  improvements $\times2.57$ and $\times2.28$  in the total number of gates, respectively.

\end{abstract}

\section{Introduction}\label{sec:intro}

Unitary synthesis is a fundamental problem in quantum computing~\cite{nielsen2000quantum}. Given a target unitary realizing a desired polynomial-size quantum computation, one needs to decompose it (either exactly or approximately) using a fault-tolerant gate set. 
Naturally, it is desirable to find a circuit realization that uses fewer operations, as this can greatly reduce errors and total runtime~\cite{fowler2012surface,litinski2019game}. The task of finding a decomposition with the least number of gates, known as optimal synthesis, is however a hard problem in general. Its precise complexity has not been pinned down yet, but it has been shown that minimizing the $T$-count using the Clifford$+T$ gate set is NP-hard (lower bound) and contained in $\operatorname{NP}^{\rm NQP}$ (upper bound)~\cite{van2023optimising}. Thus, finding optimal circuits is generally hopeless when the number of qubits grows beyond just a handful~\cite{gouzien2025provably}. 

This has motivated the study of the synthesis problem for restricted classes of unitaries, where additional algebraic structure can be exploited. 
For instance, an impressive series of results has led to optimal single-qubit unitary compilers~\cite{matsumoto2008representation,giles2013remarks,kliuchnikov2013asymptotically,ross2014optimal,gridsynth,kliuchnikov2023shorter,yamamoto2026pygridsynth,pygridsynth,morisaki2025optimal}. Clifford circuit synthesis has also been extensively studied~\cite{maslov2007linear,maslov2018shorter,bravyi2021clifford,bravyi2021hadamard,bravyi20226,peham2023depth,acuaviva2026cnot,li2026vanilla,keller2026variant}, and the decision version of optimal synthesis for circuits consisting of CNOT gates  or CZ gates between arbitrary pairs of qubits has been proven to be NP-complete~\cite{acuaviva2026cnot,li2026vanilla,keller2026variant}. More recently, another important group of unitaries whose synthesis was addressed is that of matchgate or free-fermionic unitaries~\cite{casas2025matchgate}. Here, the complexity of optimal synthesis remains open.

\begin{figure}
    \centering
    \includegraphics[width=1\linewidth]{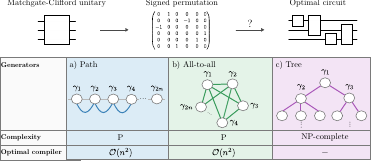}
    \caption{{\bf Overview of our main results.} We consider the problem of optimal unitary synthesis in the matchgate-Clifford group. Representing the target $n$-qubit unitary as a $2n\times 2n$ signed permutation matrix with determinant $+1$, we ask for an optimal circuit with a minimum number of Majorana braids along the edges of a given connectivity graph. We show that for path and complete graphs, the problem can be solved efficiently on a classical computer by providing explicit algorithms that run in  $\OC\left(n^2\right)$ time. In addition, we prove that the problem becomes NP-complete when the connectivity graph is an arbitrary tree.}
    \label{fig:Fig1}
\end{figure}

Importantly, the ultimate goal is not to compile an isolated single-qubit, Clifford or matchgate unitary, but rather to treat the circuits as chunks of a larger quantum computation. Indeed, any unitary can be approximated to arbitrary precision using layers of Clifford gates interleaved with layers of $T$ gates~\cite{nielsen2000quantum}. Similarly, any fermionic unitary can be factorized into layers of matchgates interleaved with layers of interacting operations~\cite{oszmaniec2017universal}.  The mathematical structure of the Clifford and matchgate unitaries can then be leveraged to exponentially reduce the size of the matrices to be synthesized,  potentially allowing for significant savings in the gate counts of the corresponding layers. In turn, this would reduce the gate count of the entire computation as well.

In this work, we focus on the group of matchgate-Clifford unitaries~\cite{bravyi2005universal}, i.e., the intersection of the matchgate and Clifford groups (see Section~\ref{sec:preliminaries} for details). This group forms a matchgate $3$-design~\cite{wan2022matchgate,heyraud2024unified} and can therefore be used in randomized fermionic protocols such as fermionic classical shadows~\cite{wan2022matchgate,heyraud2024unified,west2026particle}, benchmarking of quantum devices~\cite{helsen2022matchgate,carrasco2024gaining,portik2025clifford}, or fermionic sampling~\cite{oszmaniec2022fermion}. Moreover, it plays a crucial role in matchgate synthesis, analogous to that played by the Clifford group in general unitary compilation. Namely, the matchgate-Clifford group together with the $T$ gate (up to a phase) is universal for matchgates~\cite{casas2025matchgate}. Hence the motivation to study optimal synthesis in this group. 

At the mathematical level, the adjoint representation of $n$-qubit matchgate-Clifford unitaries is isomorphic to the group of $2n\times 2n$ signed permutation matrices with determinant $+1$ acting on the vector space of Majorana operators. Optimal matchgate-Clifford synthesis can therefore be formulated as an optimal word problem in this group, and consequently, our results may find application in the field of combinatorial group theory  beyond the quantum setting~\cite{lyndon2001combinatorial}.

We consider gate sets consisting of Majorana braids and their inverses~\cite{bravyi2005universal} on different connectivity graphs, where a braid acts as a signed transposition between pairs of Majoranas connected by an edge. The optimal synthesis problem then asks to realize the desired signed permutation using the smallest number of braids along the edges of the connectivity graph. Our main results are summarized in Fig.~\ref{fig:Fig1}. For both path and complete graphs, we prove that optimal synthesis can be solved classically in $\OC\left(n^2\right)$ time. In particular, we find efficient algorithms for optimally assigning the signs, which combined with well-known sorting algorithms for permutation matrices~\cite{miltzow2016approximation} provide explicit optimal compilers. In addition, for path graphs we investigate the optimal two-qubit-gate count, bound it, and prove that there exists a circuit simultaneously achieving this quantity and optimal total gate count.

Furthermore, we benchmark our path-graph optimal compiler on chains of  up to $n=80$ qubits against those of \texttt{Qiskit}~\cite{javadi2024quantum,qiskit} and \texttt{Tket}~\cite{sivarajah2020tket}, obtaining circuits with constant factor improvements  $\times2.57$ and $\times2.28$  in the total number of gates, respectively. We choose to benchmark the path graph compiler because the generators are simply single-qubit $S$ gates (up to a phase) and two-qubit $R_{xx}\left(\frac{\pi}{2}\right)$ rotations, under the Jordan-Wigner mapping between fermions and qubits. Thus, this case is the most relevant for matchgate circuits in practice~\cite{casas2025matchgate}~\footnote{We remark that this may no longer be true under a different fermion-to-qubit mapping~\cite{bravyi2002fermionic,verstraete2005mapping,jiang2020optimal}.}. 
We nonetheless emphasize that, alternatively to the Jordan-Wigner picture where one works with qubits and matchgates, one  could instead consider a genuinely fermionic device where many different connectivities between Majorana modes may be possible in principle~\cite{bravyi2002fermionic,alicea2011non}. In such context, our results provide guidance as to which topologies may be more beneficial in practice.

Next, we study arbitrary tree connectivity and show that the decision version of the optimal synthesis problem is NP-complete. The complexity is inherited from the NP-completeness of  token swapping on trees~\cite{aichholzer2021hardness}, which asks to minimize the number of swaps along the edges of an arbitrary tree to realize a given permutation of the tokens (in our case, the tokens would be the Majorana operators, except that we must also account for the signs). Finally, by providing a polynomial upper bound on the number of Majorana braids needed to synthesize a matchgate-Clifford unitary on any connected graph, we find that the synthesis problem on arbitrary connectivities is in NP and thus it is also NP-complete (since it is already complete for trees).

\section{Preliminaries}\label{sec:preliminaries}

In this section, we introduce the concepts that will be used in the following sections and set up notation. 

\subsection{The matchgate group}
The matchgate group on $n$ qubits is generated by single-qubit $Z$ and two-qubit $XX$ rotations with nearest-neighbor connectivity in 1D with open boundaries,
\begin{equation}
    R_z^{q}(\theta)  = e^{i\theta Z_q/2}\,,
    \qquad
    q=1,\ldots,n\,,
\end{equation}
and
\begin{equation}
    R_{xx}^{q,q+1}(\theta) = e^{i\theta X_qX_{q+1}/2}\,,
    \qquad
    q=1,\ldots,n-1\,,
\end{equation}
where $Z_q$ denotes the Pauli $Z$ acting on qubit $q$ and
$X_qX_{q+1}$ the Pauli $X\otimes X$ acting on neighboring qubits
$q$ and $q+1$. This group of unitaries is isomorphic to the Lie group $\mathbb{SPIN}(2n)$, and we recall that since the parity operator, $Z^{\otimes n}$, commutes with the rotation generators $Z_q$  and $X_qX_{q+1}$, matchgate unitaries are parity-preserving. Under the Jordan-Wigner transformation between qubits and fermions, matchgate dynamics are equivalent to free-fermionic unitaries in second quantization~\cite{knill2001fermionic,terhal2002classical,divincenzo2005fermionic}. For this reason, matchgates appear naturally in fermionic simulation primitives~\cite{kivlichan2018quantum,arute2020hartree,arrazola2022universal,verstraete2009quantum,kraus2011compressed,cervera2018exact,jiang2018quantum,dallaire2019low,sopena2022algebraic,ruiz2024bethe,ruiz2024efficient,kokcu2022fixed,kokcu2022algebraic},  where they are associated to the non-interacting part of the Hamiltonian (e.g., after Trotterization). Also known as fermionic Gaussian unitaries, they have been extensively studied in quantum computing, quantum information and quantum simulation~\cite{valiant2001quantum,knill2001fermionic,terhal2002classical,bravyi2004lagrangian,bravyi2005universal,divincenzo2005fermionic,kivlichan2018quantum,arute2020hartree,arrazola2022universal,verstraete2009quantum,kraus2011compressed,cervera2018exact,jiang2018quantum,dallaire2019low,sopena2022algebraic,ruiz2024bethe,ruiz2024efficient,kokcu2022fixed,kokcu2022algebraic,jozsa2008matchgates,jozsa2010matchgate,ramelow2010matchgate,brod2011extending,brod2014computational,brod2016efficient,oszmaniec2017universal,zhao2021fermionic,wan2022matchgate,helsen2022matchgate,oszmaniec2022fermion,diaz2023showcasing,carrasco2024gaining,mele2024efficient,stroeks2024solving,heyraud2024unified,braccia2025optimal,portik2025clifford,sierant2026theory,braccia2026commutant,lastres2026geometry,west2026particle,west2026fermionic,west2026fermionicgaussianity}.

For an $n$-mode fermionic system, we introduce the corresponding $2n$ Majorana operators,
\begin{equation}
    \gamma_1,\,\gamma_2,\,\ldots\,,\,\gamma_{2n}\,,
\end{equation}
which satisfy the canonical anti-commutation relations $\{\gamma_a,\gamma_b\}=2\delta_{ab} \id$. In Jordan-Wigner form, they are simply the following anti-commuting set of Paulis,
\begin{align}\nonumber
    \gamma_1&=XI\cdots I\,,\; \gamma_3= ZXI\cdots I\,,\quad \;\dots,\; \gamma_{2n-1}=Z\cdots\, Z X\,, \nonumber\\
        \gamma_2&=YI\cdots\, I\,,\; \gamma_4= ZYI\cdots I\,,\quad \;\dots, \;\; \gamma_{2n}\;\;\;=Z\cdots Z Y\,.
\end{align}
Importantly, all matchgate unitaries can be expressed as the exponential of a Hamiltonian that is strictly quadratic in the Majorana operators, and vice versa.  
The adjoint action of a matchgate unitary $U$ is linear on the vector space of Majorana operators, and it is specified by
\begin{equation}\label{eq:conjugation_majorana}
    U\gamma_a U^\dagger = \sum_{b=1}^{2n} R(U)_{ba}\gamma_b\,,
\end{equation}
where $R(U)\in \mathbb{SO}(2n)$ is a $2n\times 2n$ special orthogonal matrix~\cite{jozsa2008matchgates}. Notice that both $U$ and $-U$ are associated to the same orthogonal matrix (i.e., $R(U)=R(-U)$), since $\SPIN(2n)$ is the double cover of $\SO(2n)$~\cite{guaita2024representation}. Working with $R(U)$ in the standard representation of $\SO(2n)$ allows to bypass the exponential size of the matchgate unitaries $U$, which is crucial to properly address the complexity of optimally synthesizing them~\cite{casas2025matchgate}.

\subsection{The matchgate-Clifford group}

We are interested in the intersection of the matchgate and Clifford groups. We recall that  the Clifford group is the normalizer of the Pauli group in the unitary group (up to global phases), i.e., Clifford unitaries  map elements of the Pauli group to elements of the Pauli group under conjugation. Since each Majorana operator is itself a Pauli string in the qubit picture, the adjoint action of a Clifford unitary must map every Majorana operator to another Pauli operator. If $U$ is also matchgate, then the Heisenberg-evolved Majorana operator must remain within the Majorana subspace as per Eq.~\eqref{eq:conjugation_majorana}. Therefore, when $U$ is a matchgate-Clifford unitary, we find that
\begin{equation} \label{eq:matchgate-Clifford-adjoint}
    U\gamma_{a}U^\dagger = \pm \gamma_{\pi(a)}\,,
\end{equation}
for some permutation $\pi$ in the symmetric group $S_{2n}$. Consequently, the matrix $R(U)$ has exactly one nonzero entry in each row and column, equal to $\pm1$, and is a signed permutation matrix with determinant equal to $+1$. 
The group of signed permutation matrices is isomorphic to the hyperoctahedral group~\footnote{The group of symmetries of the hypercube.}, and we denote it by $B_{2n}$. Consequently, the matchgate-Clifford adjoint action is isomorphic to $B_{2n}\cap \SO(2n)$, and has order $2^{2n-1}(2n)!$.

Formally, a signed permutation of $[2n]=\{1,\dots,2n\}$  is a bijection $\tilde{w}$ of
$\{\pm1,\dots,\pm 2n\}$ with $\tilde w(-i)=-\tilde w(i)$, and so we write $\tilde w=[\tilde w(1),\dots,\tilde w(2n)]$. Considering its matrix form $w$, we have 
\begin{equation}
    w \,\vec{e}_i = \sgn(\tilde w(i))\, \vec{e}_{|\tilde w(i)|}\,,
\end{equation}
where $w\in B_{2n}$ and $\vec{e}_i$ are the canonical-basis vectors in $\mathbb{R}^{2n}$. The underlying unsigned permutation is $\pi:=|\tilde w|=\left[|\tilde w(1)|,\ldots, |\tilde w(2n)|\right]\in S_{2n}$ (see Eq.~\eqref{eq:matchgate-Clifford-adjoint}). We also denote 
\begin{equation}
    \nneg(w):=\#\{i:\tilde w(i)<0\}\,,
\end{equation}
which counts the number of negative signs in $w$, and it is related to the determinant by the identity
\begin{equation} \label{eq:determinant}
    \det(w) = \sgn(\pi)(-1)^{\nneg(w)}\,,
\end{equation}
where $\sgn(\pi)$ is the parity of $\pi$~\footnote{A permutation $\pi$ is said to be even if $\sgn(\pi)=+1$ and odd if $\sgn(\pi)=-1$.}. The parity in turn is given by
\begin{equation}
\sgn(\pi)=(-1)^{\inv(\pi)}\,,
\end{equation}
with $\inv(\pi)$ the number of inversions of $\pi$, which we define now.

\begin{definition}[Inversion number]
Let $\pi\in S_{2n}$ be a permutation. An inversion of $\pi$ is a pair
$(i,j)$, with $1\leq i<j\leq 2n$, such that $\pi(i)>\pi(j)$. The inversion
number of $\pi$ is defined as
\begin{equation}
    \inv(\pi) :=
    \#\{(i,j):1\leq i<j\leq 2n\,,\ \pi(i)>\pi(j)\}\,.
\end{equation}
\end{definition}

For example, let us consider $\pi = [3,2,1,4]$. We have that $\inv(\pi)= 3$ because $\pi(1)> \pi(2)$, $\pi(1)>\pi(3)$ and $\pi(2)> \pi(3)$. 

Therefore, in order for a signed permutation matrix to belong to $B_{2n}\cap \SO(2n)$ (i.e., have $\det(w) = +1$), it must be such that the sum of its inversion number and the number of negative entries is even. Those will be the target matrices for which we will study the problem of optimal synthesis.

\subsection{Optimal matchgate-Clifford synthesis}
\label{sec:generators}
Given a generating set $\mathcal G\subseteq  B_{2n} \cap \SO(2n)$, and a target
$w\in  B_{2n} \cap \SO(2n)$, we define
\begin{equation}
    \ell_{\mathcal G}(w) :=  \min_{g_i\in\mathcal G}\left\{
    \ell :\; w=g_1g_2\cdots g_\ell
    \right\}\,.
\end{equation}
The quantity $ \ell_{\mathcal G}(w)$ is therefore the minimum number of gates from $\mathcal{G}$ needed to synthesize $w$. The optimal synthesis problem consists in finding a word achieving
$\ell_{\mathcal G}(w)$.  The decision version of the problem is obtained by asking whether there exists a word $w$ achieving $\ell_{\mathcal G}\leq k$ for some given $k$. Clearly, the decision version belongs to the complexity class NP (Non-deterministic Polynomial time) when $\ell\in\mathcal{O}(\operatorname{poly}(n))$, because given a witness sequence $g_1g_2\cdots g_\ell$, it is efficient to check whether it actually implements $w$ by simply multiplying the matrices. However, as we show in the next section, the complexity of the problem turns out to depend strongly on the allowed generators.

\section{Results}
We study the complexity of optimal matchgate-Clifford
synthesis for different choices of generating sets. We consider gate sets consisting of Majorana braids on different connectivity graphs. That is, unitaries of the form
\begin{equation}
    U = e^{\pm\frac{\pi}{4} \gamma_i\gamma_j}\,,
\end{equation}
with $(i,j)$ an edge in the connectivity graph. The adjoint action of the braids on Majorana operators correspond to signed transpositions, given by permutation (sub-)matrices on the $i,j$ plane,
\begin{equation} \label{eq:transpositions}
    R(U)[i,j\,; i,j] = \pm\begin{pmatrix}
        0 & 1 \\ -1 & 0 
    \end{pmatrix}\,.
\end{equation}
These gates have the important property that  $R(U)^2[i,j\,; i,j]=-\id$.

\medskip

We consider one-dimensional, all-to-all and tree connectivities, and we prove the following results.

\begin{itemize}
    \item On a path graph, we show in Section~\ref{sec:path} that the optimal synthesis problem can be solved classically in $\mathcal{O}\left(n^2\right)$ time, and provide a closed-form expression for the optimal word length.  We further benchmark our optimal synthesis algorithm against other standard compilers on up to $n=80$ qubits in Section~\ref{sec:benchmark}, finding circuits with constant factors improvements ranging from $\sim\times1.5$ and $\sim\times3$ in the number of total gates and two-qubit gates. Finally, in Section~\ref{sec:two_qubits} we  derive bounds on the minimum number of two-qubit gates required to optimally synthesize  a given matchgate-Clifford unitary.

    \item On a complete graph, we show in Section~\ref{sec:complete} that the optimal synthesis can also be solved classically in $\mathcal{O}\left(n^2\right)$ time, and provide a closed-form expression for the optimal word length.  
    
    \item On tree graphs, we show in Section~\ref{sec:tree} that the decision version of the problem is NP-complete. As a corollary, we show that the problem is NP-complete on general connected graphs in Section~\ref{sec:connected}.

\end{itemize}

\subsection{Path graph}
\label{sec:path}

We first consider the physically motivated case when the generators
correspond to matchgate-Clifford operations on nearest-neighboring qubits on an open path. A natural choice for the generators is given by the single-qubit phase gate $S_q$ acting on the $q$-th qubit, and
the two-qubit gate $R_{xx}^{q,q+1}\left(\frac{\pi}{2}\right)$ acting on qubits
$q$ and $q+1$. In matrix form,
\begin{equation}
    S= R_z \left(-\frac{\pi}{2}\right) \propto
    \begin{pmatrix}
        1 & 0\\
        0 & i
    \end{pmatrix},
    \qquad
    R_{xx}\left(\frac{\pi}{2}\right)
    =
    \frac{1}{\sqrt{2}}
    \begin{pmatrix}
        1 & 0 & 0 & i\\
        0 & 1 & i & 0\\
        0 & i & 1 & 0\\
        i & 0 & 0 & 1
    \end{pmatrix}.
\end{equation}
These gates generate the matchgate-Clifford group~\cite{casas2025matchgate} and are native in some important hardware platforms such as neutral atoms~\cite{mitra2020robust} and trapped ions~\cite{blumel2021power}. 
Under the adjoint action introduced in Eq.~\eqref{eq:conjugation_majorana}, the associated orthogonal matrices $R(S_q^\dagger)$ and $R\left(R_{xx}^{q,q+1}\left(\frac{\pi}{2}\right)\right)$ act as $\frac{\pi}{2}$ rotations (i.e., signed transpositions)
of two adjacent Majorana modes,
\begin{equation}
    \gamma_a \mapsto -\gamma_{a+1}\,,
    \qquad
    \gamma_{a+1}\mapsto \gamma_a\,,
\end{equation}
while leaving all other Majorana operators unchanged (see Eq.~\eqref{eq:transpositions}). 
We denote the corresponding signed permutation matrices by $r_a$. More precisely,
\begin{equation}\label{eq:ra_def}
    r_a :=
    \begin{cases}
        R(S_q^\dagger)\,, & a=2q-1\,, \quad q\in[n]\,.\\[1mm]
        R\!\left(R_{xx}^{q,q+1}(\pi/2)\right)\,, & a=2q\,, \quad\quad\;\; q\in[n-1]\,.
    \end{cases}
\end{equation}
At the level of the abstract signed permutations,
\begin{equation} \label{eq:r_a}
     \tilde r_a =
    [1,\ldots,a-1,-(a+1),a,a+2,\ldots,2n]\,,
\end{equation}
while its inverse is 
\begin{equation}
   \tilde r_a^{-1}=
    [1,\ldots,a-1,a+1,-a,a+2,\ldots,2n]\,.
\end{equation}
\medskip

Now we consider the optimal synthesis problem with the following set of generators, 
\begin{equation}\label{eq:generators_path_garph}
    \RC:= \left\{r_a\,,\, r_a^{-1}\right\}_{a = 1}^{2n-1}\,.
\end{equation}
We find that to obtain the minimum length to synthesize a given signed permutation $w\in  B_{2n} \cap \SO(2n)$ using gates from $\RC$, the following quantity plays a crucial role.
\begin{definition}
    Given $w\in  B_{2n} \cap \SO(2n)$, we define
    \begin{equation}
        \Delta_c(w):= \det (w[1,\ldots,c \, ; 1,\ldots,c])\,, \qquad 1\leq c < 2n\,,
    \end{equation}
    where $w[1,\ldots,c \, ; 1,\ldots,c]$ denotes the $c\times c$ submatrix formed by the
first $c$ rows and first $c$ columns of $w$.     
    We say that $c$ is a cut point of $w$ if 
    \begin{equation}
        \Delta_c(w)\neq 0\,.
    \end{equation}
    If, in addition, $\Delta_c(w) = -1$, we say that $c$ is a negative defect or bad cut. 
\end{definition}

Equivalently, $c$ is a cut point of $w$ whenever $\pi([c]) = [c]$, being $\pi$ the underlying permutation in $w$. Cut points are the indices of the final rows of each block $w_k$ in the maximal direct sum decomposition of the permutation matrix in the (fixed) canonical basis,
    \begin{equation} \label{eq:permutation-blocks}
        w = \bigoplus_k w_k\,. 
    \end{equation}
    Notice that each $w_k$ is itself a (sub-)permutation matrix, although the decomposition in Eq.~\eqref{eq:permutation-blocks} is not irreducible in general. This is because the irreducible decomposition is given by the cycle decomposition, so unless every cycle only involves consecutive integers, there is at least one block $w_k$ which is reducible. However, for our purposes, this decomposition turns out to be the most useful. For instance, it is easy to see that the negative defects appear for values of $c$ such that the partial product of the block determinants $\det(w_k)$ is equal to $-1$.

    \medskip

    After these considerations, we can state our first main result, proven in Appendix~\ref{ap:proof_main_thm}.

\begin{theorem}\label{thm:path-optimal}
    Consider the generating set $\RC= \left\{r_a, r_a^{-1}\right\}_{a = 1}^{2n-1}$ and the potential
    \begin{equation}\label{eq:phi(w)_theorem}
        \Phi(w)=\inv(\pi)+2\,\nu(w)\,,\qquad
   \nu(w)=\#\{\,c\in[2n-1]:\Delta_c(w)=-1\,\}\,,
    \end{equation}
    with $\nu(w)$ the number of negative defects. Then, for every target $w\in  B_{2n} \cap \SO(2n)$, there is a classical algorithm that runs in $\OC(n^2)$ time and outputs a sequence of gates from $\RC$ that synthesizes $w$ with optimal length $\ell_{\RC}(w)=\Phi(w)$.
\end{theorem}
This theorem implies that the potential $\Phi(w)$ matches the optimal length to synthesize $w$ with the generating set $\RC$, and guarantees that a sequence of such length can be found efficiently on a classical computer. 
We recall that for an unsigned permutation $\pi$ and generating set $\GC$ given by nearest-neighbor transpositions on a path, it is well known that
$\ell_{\GC}(\pi) = \inv (\pi)$~\cite{bjorner2005combinatorics}. What we find here is that correctly assigning the signs to match those in $w$ using the gates in $\RC$ requires twice the number of negative defects extra operations.

\begin{figure}
    \centering
    \includegraphics[width=0.9\linewidth]{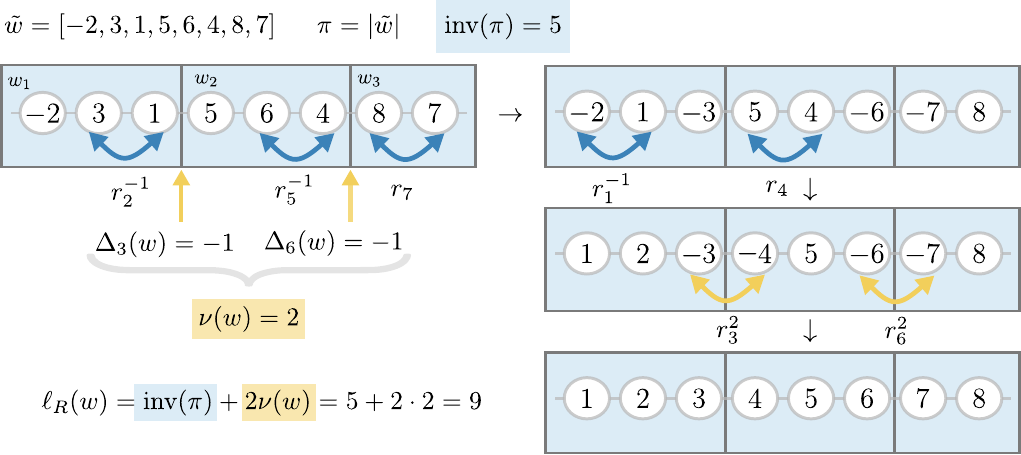}
    \caption{{\bf Example of an optimal synthesis with nearest-neighbor braids on an open path.} The number of gates from the generating set $\RC=\left\{r_a^{\pm 1}\right\}_{a=1}^{2n-1}$ required to synthesize  $w$ is $\ell _{\RC}(w)$, which splits into five inversions plus four gates to correct the sign defects captured by the number of bad cuts, $\nu(w)$. We first decompose the permutation matrix $w$ into three sub-permutation blocks $w_1, w_2$ and $w_3$. We then identify the bad cuts $c$, for which $\Delta_c = -1$. We use the inversions to synthesize the sub-permutation blocks,  leveraging the freedom in choosing $r_a$ or $r_a^{-1}$ to deposit the extra signs at the boundaries of the bad cuts, and finally use $r_c^2$ to correct them.}
    \label{fig:example_sorting}
\end{figure}

 The proof of Theorem~\ref{thm:path-optimal} consists of two steps (see Appendix~\ref{ap:proof_main_thm}). First, we prove that $\Phi(w)$ is indeed a potential by showing that a single application of a generator only changes $\Phi(w)$ by at most one, implying that $\ell_R(w)\geq \Phi(w)$. Then, we prove the converse inequality, namely $\ell_{\RC}(w)\leq \Phi(w)$, by explicitly providing an efficient algorithm that exactly matches  $\ell_R(w)= \Phi(w)$. A schematic explanation of how the algorithm works can be found in Fig.~\ref{fig:example_sorting}~\footnote{In the figure we use right multiplication for clarity, but the proof uses the equivalent left-multiplication convention.}.

When $\nu(w) = 0$, all the blocks $w_k$ have determinant $+1$ and our algorithm can correctly synthesize all the signs by using the freedom in picking either $r_a$ or $r_a^{-1}$ while sorting the blocks. On the contrary, when $\nu(w) \neq 0$, there are some blocks with $\det(w_k)=-1$. Since $r_a$ and $r_a^{-1}$ have positive determinant for all $a$, it is clear that the signs of those blocks with determinant $-1$  cannot be adjusted by using only generators that act within the block. Instead, we have to use $r_c^2$ across the cuts to correct the signs. To ensure that we only apply $2\nu(w)$ operations for these sign corrections, we place the signs at the boundaries of the blocks, such that we minimize the number of $r_c^2$ applications. This leads to the optimal $\Phi(w)$ number of total gates. For more details, see Appendix~\ref{ap:proof_main_thm}.

\subsubsection{Comparison against other compilers}
\label{sec:benchmark}
Next, we numerically benchmark our path synthesis algorithm against several alternative compilation strategies. The goal is to quantify the improvement obtained with the optimal approach relative to both matchgate-specific heuristics and general-purpose quantum compilers.

Specifically, we compare the algorithm given by Theorem~\ref{thm:path-optimal} against four alternative methods. Two of them take into account the signed-permutation structure. The first one is a sort-then-fix strategy based on bubble sort, which first realizes the underlying permutation and afterwards corrects the remaining signs. The second is a Givens-elimination strategy that operates directly on the signed-permutation matrix. This method is often used in the compilation of fermionic-shadow  settings~\cite{zhao2021fermionic,heyraud2024unified}. The general purpose compilers we employ are \texttt{Qiskit}~\cite{javadi2024quantum,qiskit} and \texttt{pytket}~\cite{sivarajah2020tket}, which synthesize the original circuits without exploiting their signed-permutation representation. 

To conduct the benchmark, we generate  matchgate-Clifford circuits as random walks of length $10n$ over the generating set $\mathcal R$.  These circuits define the target signed permutations to be synthesized. For more details, see Appendix~\ref{ap:numerics_path}. The results are shown in Fig~\ref{fig:benchmark}. In panel $(a)$, we plot the total number of gates in the compiled circuits up to $n=80$ qubits, while in panel $(b)$ the total number of two-qubit gates is depicted.
 For the total gate count, our compiler indeed produces  circuits that are shorter than all other alternatives. In particular, we find constant factor improvements ranging from $\sim 1.5$ over matchgate-dedicated methods to up to $\sim 2.6$ over general compilers. We also observe similar reductions in the number of two-qubit gates, although our algorithm is designed to minimize the total gate count and therefore does not guarantee optimality in the two-qubit count. We address this gap in the following section.

\begin{figure}[t]
\centering
\IfFileExists{figs/benchmark.pdf}%
  {\includegraphics[width=1\linewidth]{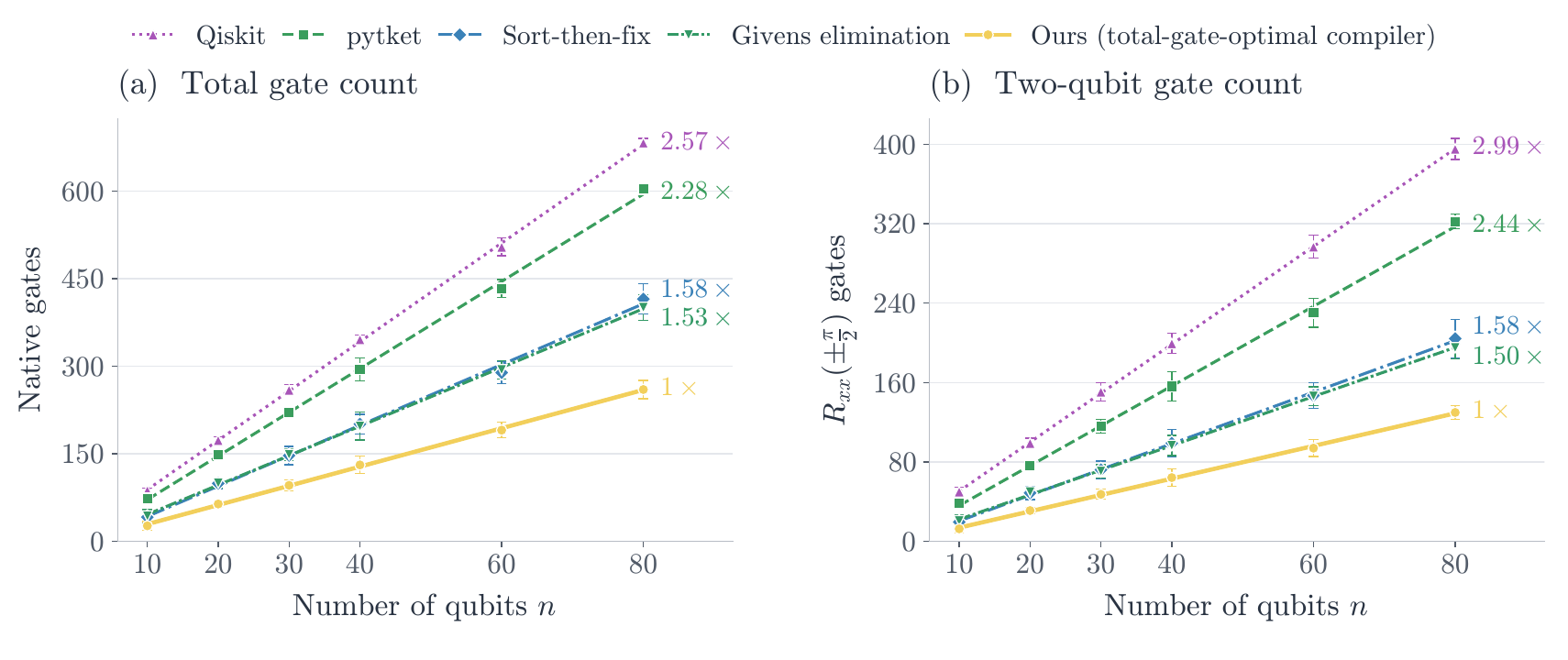}}%
  {\IfFileExists{benchmark.png}%
    {\includegraphics[width=0.6\linewidth]{benchmark.png}}%
    {\fbox{\parbox[c][3.4cm][c]{0.9\linewidth}{\centering\itshape
       Place \texttt{benchmark.pdf} or \texttt{benchmark.png} in this directory
       to display the benchmark plot here.}}}}
\caption{\textbf{Performance of our optimal path synthesis algorithm.}
Randomly generated Clifford-matchgate circuits are compiled using different strategies: our optimal compiler, a Givens elimination strategy, sort then fix, \texttt{pyket} and \texttt{Qiskit}. In panel $(a)$ we plot the average total gate count for different circuit sizes, and in panel $(b)$ we plot the average number of two-qubit gates.  Error bars represent the standard deviation. 
In both plots, we fit the data to straight lines and show the constant relative improvement achieved by our compiler over the other methods.}
\label{fig:benchmark}
\end{figure}

\subsubsection{Optimality in the number of two-qubit gates}\label{sec:two_qubits}
So far we have treated $S$ and $R_{xx}\left(\frac{\pi}{2}\right)$ gates as equally costly, since both are Clifford gates. However, in general two-qubit gates are significantly more noisy than single-qubit unitaries. Here, we investigate how much we can reduce the number of $R_{xx}\left(\frac{\pi}{2}\right)$ operations, or equivalently, the number of generators $r_a$ with even $a$ (see Eq.~\eqref{eq:ra_def}).

A useful tool for treating two-qubit gates differently is to group pairs of Majorana modes into physical modes, as
\begin{equation}\label{eq:physical_modes}
    \{1, 2\}, \{3,4\},\dots,\{2n-1, 2n\}\,. 
\end{equation}
The only generators that move Majoranas between physical modes are even-index generators. For each Majorana $\gamma_i$, we denote $s_i$  its initial mode and $t_i$ its target mode. For a boundary $\beta$ between modes
$\beta$ and $\beta+1$, we define
\begin{equation}
    F_\beta
    :=
    \#\{i:s_i\leq\beta<t_i\}\,,
\end{equation}
which is the number of Majoranas that start to the left of the boundary  $\beta$  and whose target position lies on the right of this boundary. By conservation, the same number $F_\beta$ must cross the boundary in the opposite direction as well. The total right (or left) flow over all boundaries is 
\begin{equation}\label{eq:rightward_flow}
F(\pi):=\sum_{\beta=1}^{n-1}F_\beta\,.
\end{equation}

In addition, we define $\nu_{\operatorname{even}}(w)$ as the number of bad cut points that lie at even positions,
\begin{equation}\label{eq:nu_even}
    \nu_{\operatorname{even}}(w)
    :=
    \#\left\{
        c\in\{1,\ldots,2n-1\}:
        \pi([c])=[c]\,,\;
        c\equiv0\;\ \bmod 2,\,\;
        \Delta_c(w)=-1
    \right\}\,.
\end{equation}
Then, as proven in Appendix~\ref{ap:proof_optimal_two_qubit}, the following result holds.

\begin{theorem}\label{th:optimal_two-qubit}
For every $w\in B_{2n} \cap \SO(2n)$ with underlying permutation $\pi=|w|$, the minimum number of two-qubit gates over all words realizing $w$,  $\#_2(w) $,  is given by 
\begin{equation}\label{eq:two_qubit_count}
    \#_2(w) = \mu(\pi) + 2\nu_{\operatorname{even}}(w)\,, 
\end{equation}
where $\mu(\pi)$ is  the minimum number of even-indexed generators appearing in a reduced word for $\pi$. 
The quantity $\mu(\pi)$ satisfies
\begin{equation}\label{eq:bounds_mu(pi)}
    F(\pi)
    \leq
    \mu(\pi)
    \leq
    2F(\pi)\,, \quad F(\pi) \leq \left\lfloor\frac{n^2}{2}\right\rfloor\,,
\end{equation}
where $F(\pi)$ is the net right (or left) flow of Majoranas over all physical boundaries. The upper bound is constructive and therefore gives a factor-$2$ approximation to $\mu(\pi)$.

\medskip

Moreover, there exists a word that simultaneously achieves the minimum total length
\begin{equation}
    \ell_R(w) = \inv(\pi)+2\nu(w)\,,
\end{equation}
as given by Theorem~\ref{thm:path-optimal}, and the minimum two-qubit gate count $\#_2(w)$.

\end{theorem}

Let us discuss the implications of Theorem~\ref{th:optimal_two-qubit}. First, Eq.~\eqref{eq:two_qubit_count} indicates that, when minimizing the total number of two-qubit gates, we can split the contribution from the sign correction from the number of two-qubit gates to build the underlying permutation $\pi$.
The sign correction cost is fixed and given by $2\nu_{\operatorname{even}}(w)$. The remaining contribution depends only on the underlying permutation $\pi$ and is captured by $\mu(\pi)$, the minimum number of even-indexed generators among reduced words for $\pi$. Thus, the nontrivial part of minimizing the two-qubit gate count is determining $\mu(\pi)$.

By definition, a reduced word attaining $\mu(\pi)$ has total length $\inv(\pi)$. However, once the sign corrections required to realize the full signed permutation $w$ are included, it is not a priori obvious that minimizing the two-qubit gate count is compatible with minimizing the total synthesis length. Theorem~\ref{th:optimal_two-qubit} shows that these two objectives can in fact be achieved simultaneously.

The complexity in finding the optimal two-qubit gate count therefore lies in $\mu(\pi)$. Whether $\mu(\pi)$ can be computed, and whether a word attaining this quantity can be generated efficiently, remains open. However, we show that this quantity admits efficiently computable upper and lower bounds, given by Eq.~\eqref{eq:bounds_mu(pi)}. This gives a factor-$2$ approximation to the optimal quantity $\mu(\pi)$. 

Importantly, the upper bound for $\mu(\pi)$ is constructive, and therefore gives an approximation to the optimal two-qubit gate word.  For this construction, we build a word $W'$ for $\pi$ containing at most $2F(\pi)$ even-indexed generators by routing the Majoranas across the mode boundaries. Then, we use the Coxeter deletion property and remove a pair of generators without increasing the number of even-indexed generators. Therefore,
\begin{equation}
    E(W') \leq 2F(\pi) \leq 2\mu(\pi),
\end{equation}
where the last inequality follows from $\mu(\pi)\geq F(\pi)$. Hence the constructed word $W'$ provides a factor-$2$ approximation to the minimum number $\mu(\pi)$ of even-indexed generators. For more details on this routing, see Appendix~\ref{ap:proof_optimal_two_qubit}.

Although the Coxeter deletion property does not specify which generators should be deleted, we can find a redundant pair of generators by iterating over them. We delete such a pair when it does not change the underlying permutation, and we repeat until there are $\inv(\pi)$ generators left, which results in a reduced word. This procedure can be executed in polynomial time.  The argument does not extend to signed permutations directly, but since we are only interested in $\mu(\pi)$ and it only involves the underlying permutation, this is not an issue.

\subsection{Complete graph}
\label{sec:complete}
In this section, we consider a different connectivity graph and prove that optimal synthesis can be classically solved efficiently for generators with all-to-all connectivity. The generating set is given by 
\begin{equation}\label{eq:generating_all}
    \RC_{\mathrm{all}} :=
    \left\{
        r_{ab}\,,r_{ab}^{-1} : \;1\leq a<b\leq 2n
    \right\}\,.
\end{equation}
These generators act on the
two-dimensional subspace spanned by the canonical basis vectors $\vec{e}_a,\vec{e}_b$ as
\begin{equation}
    r_{ab} \vec{e}_a=-\vec{e}_b\,,
    \qquad
    r_{ab} \vec{e}_b=\vec{e}_a\,,
\end{equation}
while other basis elements are left unchanged (see Eq.~\eqref{eq:transpositions}). At the unitary level, $r_{ab}$ is realized by the rotation
\begin{equation}
    U_{ab} = \pm e^{\left(\frac{\pi}{4}\gamma_a\gamma_b\right)}\,.
\end{equation}
This unitary, for spatially separated modes, it is generally not a local two-qubit gate in the Jordan-Wigner representation, even though it admits a matchgate decomposition in terms of one- and two-qubit gates. This means that we are including non-local matchgates in our generating set. These operations may nevertheless be natively available in fermionic architectures. For instance, in reconfigurable arrays where fermionic modes can be exchanged with almost no overhead~\cite{gonzalez2023fermionic}.

The optimal synthesis problem is now, given the all-to-all generating set in Eq.~\eqref{eq:generating_all}, whether we can efficiently find optimal sequences for signed permutation matrices $w\in B_{2n}\cap \SO(2n)$. 
To answer this question, we need to take into account the cycle structure of the underlying permutation.  Consider the unsigned permutation $\pi = |\tilde w|$ and its cycle decomposition,
\begin{equation}
    \pi = C_1 C_2\dots C_{\kappa(\pi)}\,,
\end{equation}
where each $C_i$ is a cycle of the permutation $\pi$ and $\kappa(\pi)$ is the number of cycles, including fixed points. 
Now consider one cycle $C= (c_1, c_2, \dots, c_L)$ of length $L$. In the signed permutation $\tilde w$, there is a $\pm1$ sign for each element, 
\begin{equation}
  \tilde  w(c_1) = \epsilon_1 c_2\,, \quad \tilde w(c_2) = \epsilon_2 c_3\,, \quad \ldots\,. \quad \tilde w(c_L) = \epsilon_L c_1,\,
\end{equation}
with $\epsilon_j\in \{\pm 1\}$. We define the sign of the cycle as $S(C) = \prod_{j = 1}^L \epsilon_j$. The determinant associated to the cycle $C$ of the signed permutation $w$ is given by 
\begin{equation}\label{eq:determinant_cycle}
    \delta_w(C):= \det w[C;C] = (-1)^{L-1} S(C)\,,
\end{equation}
where $w[C;C]$ is the permutation matrix restricted to the cycle coordinates. 
We say that a cycle is defective whenever $\delta_w(C) = -1$. Equivalently, we have a defective cycle when $ S(C) = (-1)^L$.  Crucially, notice that the number of defective cycles must be even since the total determinant is equal to $+1$.

Then, as proven in Appendix~\ref{ap:proof-all-to-all}, the following result holds.
\begin{theorem}
\label{thm:all-to-all}
For every signed permutation $w\in B_{2n}\cap \SO(2n)$ with underlying permutation $\pi=|\tilde w|$, the minimum number of generators
from $\RC_{\mathrm{all}}$ required to synthesize $w$ is
\begin{equation}
    \ell_{\mathrm{all}}(w) = \bigl(2n-\kappa(\pi)\bigr)+d_-(w)\,,
    \label{eq:all-to-all-length}
\end{equation}
where $\kappa(\pi)$ the number of cycles in $\pi$ and $d_-(w)$ is the number of defective cycles. Moreover, $\ell_{\mathrm{all}}(w)\leq 2n$, and an optimal word is computable in $\mathcal \OC\left(n^2\right)$ time on a classical computer.
\end{theorem}

\begin{figure}
    \centering
    \includegraphics[width=1\linewidth]{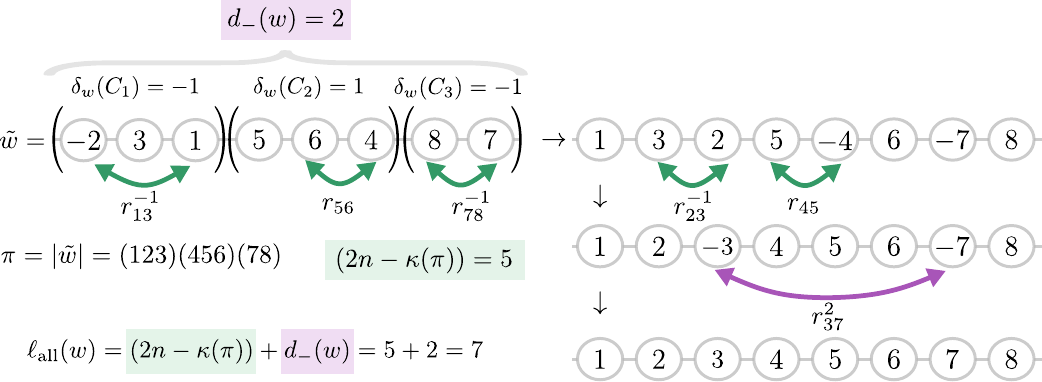}
    \caption{{\bf  Example of an optimal synthesis with all-to-all braids.}
    The underlying permutation $\pi=|\tilde w|$ is decomposed into its cycle structure, with $\kappa(\pi)=3$. Sorting these cycles requires $2n-\kappa(\pi)=5$ generators, as shown in green. The target contains $d_-(w)=2$ defective cycles, which require two additional generators for the sign correction. After the permutation has been sorted, the remaining sign defects are paired and removed by $r_{37}^2$, as shown in pink. This gives an optimal total length
$\ell_{\operatorname{all}}(w)=7$. For visual clarity, we use right multiplication convention.}
    \label{fig:example_all_to_all}
\end{figure}

The term $ \bigl(2n-\kappa(\pi)\bigr)$ in Eq.~\eqref{eq:all-to-all-length} is the cost associated to sorting the unsigned permutation $\pi$ using arbitrary swaps between Majoranas, corresponding to the classical transposition distance~\cite{diaconis1988group}. The number of residual signs that need to be corrected requires exactly $d_-(w)$ extra generators. This term plays an analogous role to $\nu(w)$ in Section~\ref{sec:path}. It takes into account the unavoidable sign defects that cannot be corrected without using generators that act non-trivially across cycles.

An example of the optimal sorting strategy is shown in Figure~\ref{fig:example_all_to_all}, where a signed permutation matrix $w$ associated to a $4$-qubit matchgate-Clifford unitary is brought to the identity by a sequence of gates (pink and green arrows). By inverting the sequence,  we obtain a word synthesizing $w$ with optimal length. The main idea is to express the underlying permutation $\pi = |\tilde w|$ in cycles and identify the defective ones (those with $\delta_w(C) = -1$). Then, we can readily sort the cycles using the minimum number of transpositions, employing the freedom in choosing between $r_{ab}$ and $r_{ab}^{-1}$ to correctly deposit the signs, up to the negative defects. Since the number of negative defects is even, we can pair the corresponding cycles and use two consecutive braids, $r_{ab}^2$, with $a$ in one defective cycle and $b$ in another, to correct all negative defects.

The closed-form expression for the optimal length also gives the diameter of the corresponding Cayley graph, that is, the maximum value of $\ell_{\mathrm{all}}(w)=2n-\kappa(\pi)+d_-(w) $ over all $w\in B_{2n}\cap\SO(2n)$. In particular,
\begin{equation}
   \ell_{\mathrm{all}} \leq 2n \,.
\end{equation}
This maximum  is attained by $-e\in  B_{2n}\cap\SO(2n)$.  Consequently, every matchgate-Clifford element can be synthesized using at most $2n$ all-to-all generators, and this bound is tight.

Finally an optimal synthesis sequence can be constructed in polynomial time. At each step, the algorithm applies a generator that decreases the distance to the identity by one. These operations require at most $\mathcal O(n)$ steps (identify defective cycles and sort the cycles and signs). Since the Cayley-graph diameter is $2n$, at most $2n$ of such iterations are required. The complete synthesis procedure can therefore be implemented in $\mathcal O\left(n^2\right)$ time.

\subsection{Tree graphs}
\label{sec:tree}

We now consider the case when the connectivity graph is a tree given as input to the synthesis problem (in other words, when the connectivity graph is an arbitrary tree). We prove the following result, see Appendix~\ref{ap:trees} for details.

\begin{theorem} \label{th:tree-npcomplete}
    The decision version of gate-optimal matchgate-Clifford synthesis with Majorana braids on tree connectivity graphs is NP-complete.
\end{theorem}

This result showcases that the complexity of the optimal synthesis problem heavily depends on the chosen generating set. Indeed, the previous sections explained how optimal synthesis can be achieved with polynomial-time classical computation when the connectivity graph is an open path or a complete graph. In contrast, we here show that  when the graph is an arbitrary tree, the complexity of the problem jumps to NP-complete, for which no worse-case polynomial-time algorithm exist unless P=NP. The reduction is from the well-known routing problem, which is NP-complete on trees~\cite{aichholzer2021hardness} and asks to minimize the number of swaps in a given graph required to realize a prescribed permutation of the tokens placed on its vertices. This problem is very similar but not completely equivalent to optimal matchgate-Clifford synthesis, and the difference resides in the signs introduced by the Majorana braids. We prove that for every permutation, it is always possible to efficiently find a configuration of signs such that the minimum length of a braid sequence synthesizing the signed permutation is exactly the minimum number of swaps required for the original permutation. Consequently, if an algorithm solves the signed permutation synthesis problem, it also solves the routing problem. Importantly, notice that, for the reduction to work in this direction, we have to determine the sign configuration from the permutation itself, without prior knowledge of an  optimal routing. We furthermore provide a polynomial upper bound on the length of a sequence synthesizing a given unit-determinant signed permutation on a tree, which guarantees that the problem is indeed in NP and thus it is NP-complete.

In addition, we show that optimizing the depth in matchgate-Clifford circuit synthesis with tree connectivity is also NP-complete, again using the results in Ref.~\cite{aichholzer2021hardness} and the same idea as before to handle the signs (see Appendix~\ref{ap:trees} for a proof).

\begin{theorem} \label{th:tree-depth}
    The decision version of depth-optimal matchgate-Clifford synthesis with Majorana braids on tree connectivity graphs is NP-complete.
\end{theorem}

\subsection{Connected graphs}\label{sec:connected}

A fairly direct corollary of Theorems~\ref{th:tree-npcomplete} and~\ref{th:tree-depth} is that when we are allowed to pick arbitrary Majorana braids as generators (equivalently, choose arbitrary connectivity graphs, that are assumed to be connected), the optimal synthesis problem in the matchgate-Clifford group is NP-complete. Since we have already shown that optimal matchgate-Clifford synthesis is NP-complete for tree connectivity, it suffices to provide a polynomial upper bound on the length of sequences synthesizing signed permutations on any graph in terms of Majorana braids. We do so in Appendix~\ref{ap:general-graphs}, and hence arrive at the following corollary.

\begin{corollary} \label{cor:connected-graphs}
    The decision version of gate-optimal and depth-optimal matchgate-Clifford synthesis with Majorana braids on general connectivity graphs is NP-complete.
\end{corollary}

\section{Discussion}
\label{sec:disc}

In this work, we have explored optimal exact synthesis within the matchgate-Clifford group. In particular, we have derived two closed-form expressions for the optimal synthesis length in terms of Majorana braids with path and complete-graph connectivity, together with efficient optimal synthesis algorithms. 
We have additionally benchmarked our path-synthesis algorithm against standard quantum compilers on chains of up to $n = 80$ qubits, and derived bounds on the minimum number of two-qubit gates. 
Moreover, we have shown that when an arbitrary tree connectivity graph is supplied as part of the input, the decision version of the optimal synthesis problem becomes NP-complete.

These results showcase that the complexity of optimal synthesis strongly depends  on how the generating set, and, more specifically, its connectivity, enter the problem. 
Identifying fixed connectivities that admit efficient optimal synthesis thus remains an interesting open question. Overall, our results help to better understand the synthesis of an important group of unitaries in quantum computing and may find application in compressing circuits for fermionic simulation. This topic has received some attention lately, and recent progress on fermionic routing has shown that polylogarithmic-depth implementations of fermionic-mode permutations can be obtained~\cite{maskara2025fast,constantinides2025low}, while $\OC(\log n)$ depth synthesis of matchgate-Clifford unitaries using general Clifford circuits under all-to-all qubit connectivity was proven asymptotically optimal~\cite{west2026strong}. Here, we have taken a complementary approach by pinning down the complexity in terms of the total number of gates needed for synthesis within the matchgate-Clifford group.

More broadly, this work constitutes a step towards understanding optimal synthesis for structured subgroups of the unitary group, where the complexity landscape still remains incomplete. 
Recently, both the decision version of optimal CNOT circuit synthesis under all-to-all connectivity and of depth-optimal CZ circuit synthesis were shown to be NP-complete~\cite{acuaviva2026cnot,li2026vanilla,keller2026variant}. In contrast to Clifford circuits, the complexity of optimal synthesis for the full matchgate group (either exact or approximate) remains unknown~\cite{casas2025matchgate}. Hence, determining how algebraic structure, generating sets, and hardware connectivity combine to make exact optimal synthesis tractable or intractable constitutes an interesting research direction for future work.

To conclude, we stress that although our main motivation comes from matchgate-Clifford synthesis, our results can also be viewed as statements about the shortest word length in signed permutation groups. Indeed, the adjoint representation of the matchgate-Clifford group is isomorphic to the unit-determinant signed permutation group, $B_{2n}\cap \SO(2n)$. Standard word-length formulas are known for Coxeter groups of type $B$ and $D$~\cite{bjorner2005combinatorics} and for related rotation subgroups~\cite{brenti2007alternating}, but they apply to different generating sets, not the Majorana-braid generators considered here. In particular, the standard Coxeter length functions do not capture the additional sign-correction term. Our results therefore provide exact word-length formulas for this physically motivated generating set.

\medskip

\paragraph*{Data and code availability.} The source code and benchmark scripts to reproduce Fig.~\ref{fig:benchmark}, together with the plotted raw data, are publicly available on GitHub at \url{https://github.com/yuxuanzhang1995/matchgate-clifford-synthesis}.

\paragraph*{Acknowledgments.}
We thank  Richard K\"ueng and Luna Lima Keller for insightful conversations on the complexity of circuit synthesis of restricted groups of unitaries, and Paolo Braccia for in-depth discussions on the synthesis of (Haar-random) matchgate-Clifford unitaries. B.C. received funding from the European Union, supported by the EuroHPC Joint Undertaking and its members under the Grant Agreement Nº 101159808, including top-up funding by Ministry for Digital Transformation and the Civil Service of the Spanish Government. B.C. was also supported by the Laboratory Directed Research and Development (LDRD) program of Los Alamos National Laboratory (LANL) under project number 20260043DR. D.G.-M. acknowledges financial support from the ERC via the Starting grant q-shadows (101117138) and from the Austrian Science Fund (FWF) via the SFB BeyondC (10.55776/FG7). Y.Z. acknowledges funding from the European Research Council (ERC) under the European Union's Horizon 2020 research and innovation programme (grant agreement No.~864597).

\paragraph*{AI assistance disclosure.} We acknowledge the use of ChatGPT 5.6 and Claude to prove the main results in the paper and help develop the code used to generate Fig.~\ref{fig:benchmark}. All scientific arguments, results, and conclusions were independently reviewed, verified and simplified by the authors.

\printbibliography

\appendix
\setcounter{theorem}{0}
\setcounter{lemma}{0}
\setcounter{proposition}{0}
\setcounter{corollary}{0}

\section{Proof of Theorem~\ref{thm:path-optimal}}\label{ap:proof_main_thm}
In this section, we present a proof of Theorem~\ref{thm:path-optimal}, which we restate here for clarity. 
\begin{theorem}\label{eq:phi(w)_theorem_app}
    Consider the generating set $\RC= \left\{r_a, r_a^{-1}\right\}_{a = 1}^{2n-1}$ and the potential
    \begin{equation}\label{eq:phi(w)_theorem}
        \Phi(w)=\inv(\pi)+2\,\nu(w)\,,\qquad
   \nu(w)=\#\{\,c\in[2n-1]:\Delta_c(w)=-1\,\}\,,
    \end{equation}
    with $\nu(w)$ the number of negative defects. Then, for every target $w\in  B_{2n} \cap \SO(2n)$, there is a classical algorithm that runs in $\OC(n^2)$ time and outputs a sequence of gates from $\RC$ that synthesizes $w$ with optimal length $\ell_{\RC}(w)=\Phi(w)$.
\end{theorem}

\begin{proof}
The proof consists  of two steps. First, we provide a lower bound on the length of a sequence synthesizing $w$. Then, we prove a matching upper bound on $\ell_{\RC}(w)$. 

We start by showing that $\Phi(w)$ is a lower bound on the word length. To do so, we show that multiplication by any generator changes the potential by at most one in absolute value.
    \begin{lemma}\label{lem:potential_lower_bound}
        $\Phi(gw)-\Phi(w)\in\{-1,+1\}$ for all $w\in B_{2n} \cap \SO(2n)$ and $g\in \RC$.
    \end{lemma}
    \begin{proof}
       We have to bound the change of both contributions in $\Phi(w)$ once we apply a given generator. Let us start with the change in the number of bad cuts. 
        For a cut $c$ and $g = r_a^{\pm 1}\in \RC$, we find that $\Delta_c(gw) = \Delta_c(w)$ for all $a\neq c$. We can split this into two cases: 
        \begin{itemize}
            \item $c<a$: In this case, $\Delta_c(w)$ does not change because $g$ does not affect the sub-matrix $w[1:c, 1:c]$. 
            \item $c> a$: $\Delta_c(w)$ contains both coordinates affected by $r_a^{\pm1}$. The generator $g$ acts entirely inside the first $c$ coordinates. Consequently, $gw[1:c, 1:c] = g[1:c, 1:c]\,w[1:c, 1:c]$ (which is not true in general).  Therefore, the determinant is unchanged, 
            \begin{align}
                \det((gw)[1:c, 1:c]) = & \det(w[1:c, 1:c] )\det(g[1:c, 1:c]) \nonumber \\=& \det(w[1:c, 1:c] )= \Delta_c(w)\,,
            \end{align}
            where we have used that $\det(g[1:c, 1:c]) = 1$. 
        \end{itemize}

        We can thus see that the value of the determinant at every pre-existing  cut $c\neq a$ in $w$ is unchanged. However, we must note that when $c>a$, new cuts can appear inside $w[1:c, 1:c]$. For instance, consider $w=r_1$, which has no bad cut. Then $r_1w=r_1^2$ has one bad cut. Importantly, a new cut can only appear at $a$, because for every $d\neq a$, the transposition $r_a^{\pm1}$ leaves $[d]$ invariant, and hence
\begin{equation}
   \pi'([d])=[d]
    \quad\Longleftrightarrow\quad
    \pi([d])=[d]\,,
\end{equation}
where $\pi'$ is the underlying permutation of $ r_a^{\pm 1} w$ and similarly $\pi=|\tilde w|$.

        We next focus on the case when $a = c$ for some cut. Here, the sub-matrix $w[1:c, 1:c]$ contains coordinate $a$ but not coordinate $a+1$. It is clear that a cut point at $c=a$ cannot be preserved by the action of $r_a^{\pm1}$, since the underlying permutation satisfies $\pi([a])=[a]$ and left multiplication by $r_a^{\pm1}$ exchanges the values $a$ and $a+1$ wherever they occur in the window notation (also adding an extra sign). Hence, the application of $r_a^{\pm 1}$ necessarily removes one cut point, merging two contiguous blocks $w_k$ and $w_{k+1}$ in $w$ (recall Eq.~\eqref{eq:permutation-blocks}). This in turn implies that the number of bad cuts in this case can only change by at most one. Consequently, we have
        \begin{equation}
            \nu(gw)- \nu(w)\in \{-1,0,1\}\,.
        \end{equation}

    For the inversion number, left multiplication by an adjacent transposition $r_a^{\pm1}$ changes the inversion number by exactly one,
    \begin{equation}
        \inv(|r_a^{\pm1} w|)-\inv(|w|)\in\{-1,+1\}\,.
    \end{equation}
 Therefore, we arrive at
\begin{equation}
    \Phi(r_a^{\pm1} w)-\Phi(w)\in\{\pm1,\pm3\}\,.
\end{equation}
We now show that $\pm3$ jumps cannot occur. Let us consider the two possible cases: 
\begin{enumerate}
    \item A negative defect at $a$ is destroyed,  so $\Delta\nu=-1$. This can only happen when $a=c$ for some cut $c$. Since $\pi([a])=[a]$, it follows that $\pi^{-1}(a)\leq a< a+1\leq \pi^{-1}(a+1)$ and the transposition between $\pi^{-1}(a)$ and $\pi^{-1}(a+1)$ creates an inversion, resulting in $ \Delta\inv=+1$. Hence
    \begin{equation}
        \Delta\Phi = \Delta\inv+2\Delta\nu = 1-2=-1\,.
    \end{equation}

    \item A negative defect at $a$ is created, so $\Delta\nu=+1$. Then $a$ is a cut point of $r_a^{\pm1}w$, while it was not a cut point of $w$. Since $\pi'([a])=[a]$, we have
\begin{equation}
    \pi\left([a]\right)
    =|\tilde r_a^{\pm1}| \left([a]\right)
    =\{1,\ldots,a-1,a+1\}\,,
\end{equation}
where $|\tilde r_a^{\pm1}|$ is the underlying permutation of $r_a^{\pm1}$ and we used $|r_a^{\pm1}|^2$ is the identity permutation $e$.
Hence
\begin{equation}
    \pi^{-1}(a+1)\leq a < \pi^{-1}(a)\,.
\end{equation}
Thus, this is an inversion of $\pi$. Left multiplication by
$r_a^{\pm1}$ exchanges the values $a$ and $a+1$, removing this inversion.
Since an adjacent transposition changes the inversion number by exactly one,
$\Delta\inv=-1$. Consequently,
\begin{equation}
    \Delta\Phi=\Delta\inv+2\Delta\nu=-1+2=+1\,.
\end{equation}

 \item If $\Delta\nu=0$, then $\Delta\Phi=\Delta\inv\in\{-1,+1\}$. 
     
\end{enumerate}

Thus, in every case,
\begin{equation}
    \Phi\left(r_a^{\pm1} w\right)-\Phi(w)\in\{-1,+1\}\,,
\end{equation}
which proves the lemma. 
All the cases are summarized in Table~\ref{tab:Th1}.
       
\end{proof}  

\begin{table}[t]
\centering
\begin{tabular}{c|c|c|c}
{\bf Changed induced at} $a$
& $\Delta\nu$
& $\Delta\inv$
& $\Delta\Phi$
\\
\hline
Negative defect destroyed
& $-1$
& $+1$
& $-1$
\\
Negative defect created
& $+1$
& $-1$
& $+1$
\\
No change in $\nu$, inversion created
& $0$
& $+1$
& $+1$
\\
No change in $\nu$, inversion destroyed
& $0$
& $-1$
& $-1$
\end{tabular}
\caption{Possible changes of the potential $\Phi$ under multiplication by a
generator $r_a^{\pm1}$.}\label{tab:Th1}
\end{table}

Since $\Phi(e)=0$ and each application of a generator changes the absolute value of $\Phi$ by at most one, any word over $\RC$ that reaches $w$ from the identity must have length at least $\Phi(w)$, i.e.,
\begin{equation}
    \ell_{\RC}(w)\geq \Phi(w)\,.
\end{equation}

\medskip

Now we continue with the upper bound, which is constructive. Namely, we provide an algorithm that builds  a word of length exactly $\Phi(w)$, which both certifies $\ell_{\RC}(w)\le\Phi(w)$ and is the optimal compiler we use throughout.

First, let us recall that a signed permutation matrix $w\in B_{2n}\cap \SO(2n)$ decomposes in the fixed canonical basis of $\mathbb{R}^{2n}$ at its cut points into cut-free blocks $w_1, \dots, w_t$, being $t$ the number of such blocks (see Eq.~\eqref{eq:permutation-blocks}).  A cut point of $w$ is a value $c$ such that $\pi([c]) = [c]$, where $\pi$ is the underlying permutation of $w$. Thus, we say that $w$ is cut-free if it has no non-trivial cut, that is, there is no $c\in[2n-1]$ such that $\pi([c])=[c]$. Similarly, we say that a block $w_k$ is cut-free if $w_k$ does not contain any non-trivial cut, i.e., the only cut it contains is located at the last row of the block. 

We define the inversion graph $\Gamma_w$ as the graph with vertex set $[2n]$, where two vertices $u<v$ are connected by an edge  if $\{u,v\}$ is an inversion of $w$. We prove the following lemma.

\begin{lemma}\label{le:connected_graph}

Let $\sigma$ be a permutation in $S_{m}$. Then, $\Gamma_{\sigma}$ is connected if and only if $\sigma$ is cut-free.
\end{lemma}

\begin{proof}
First, suppose that $\sigma$ has some non-trivial cut. Then, there exists
$c\in[m-1]$ such that
\begin{equation}
    \sigma([c])=[c]\,.
\end{equation}
Therefore, every value $u\leq c$ appears before every value
$v>c$ in $\sigma$. Since $u<v$, the pair $\{u,v\}$ is not an inversion. Hence, there are no edges in $\Gamma_\sigma$ connecting the sets $[c]$ and $\{c+1,\ldots,m\}$. 
Thus, $\Gamma_\sigma$ is disconnected.

Conversely, suppose that $\Gamma_\sigma$ is disconnected. Consider the connected component containing the value $1$, and let $c$ be the maximum value in this component such that $[c]$ is in the component. 
If $c=m$, then all vertices belong to the same connected component, in contradiction with the assumption that $\Gamma_\sigma$ is disconnected.
Therefore, $c<m$, and $c+1$ does not belong to the component. 
Since there is no edge between $c+1$ and any value in $[c]$, the value $c+1$ must appear after every value less than or equal to $c$ in $\sigma$.

Moreover, suppose that a value $d>c+1$ appears before a value $u\leq c$. If this happened, then $\{u,d\}$ would be an inversion and $d$ would belong to the same connected component as $u$. But since $d>c+1$ appears before $u$, and $u$ appears before $c+1$, the pair $\{c+1,d\}$ would also be an inversion, connecting $c+1$ to that same component, which is a contradiction.

Hence, every value in $[c]$ appears before every value larger than $c$, so
\begin{equation}
    \sigma([c])=[c].
\end{equation}
Thus, $\sigma$ has a non-trivial cut.

We conclude that $\Gamma_\sigma$ is connected if and only if $\sigma$ is cut-free.
\end{proof}

Figure~\ref{fig:inversion-graphs} illustrates the connection between the presence of non-trivial cuts and the connectivity of the inversion graph with two examples.

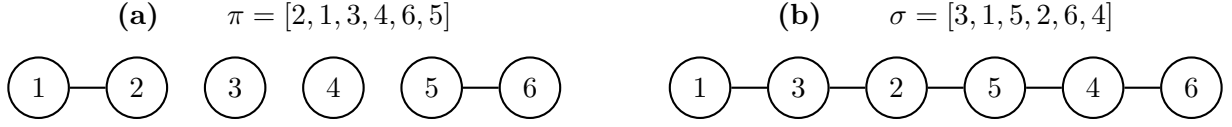
\begin{figure}[t]
    \centering

    %------------------------------------------------
    % (a) Decomposable permutation
    %------------------------------------------------
    \begin{minipage}[t]{0.47\textwidth}
        \centering
        \textbf{(a)} \qquad $\pi=[2,1,3,4,6,5]$

        \vspace{0.3cm}

        \begin{tikzpicture}[
            vertex/.style={
                circle,
                draw,
                thick,
                minimum size=8mm,
                inner sep=0pt
            },
            every path/.style={thick}
        ]

        \node[vertex] (1) at (0,0) {1};
        \node[vertex] (2) at (1.3,0) {2};
        \node[vertex] (3) at (2.6,0) {3};
        \node[vertex] (4) at (3.9,0) {4};
        \node[vertex] (5) at (5.2,0) {5};
        \node[vertex] (6) at (6.5,0) {6};

        \draw (1) -- (2);
        \draw (5) -- (6);

        \end{tikzpicture}
    \end{minipage}
    \hfill
    %------------------------------------------------
    % (b) cut-free permutation
    %------------------------------------------------
    \begin{minipage}[t]{0.47\textwidth}
        \centering
        \textbf{(b)} \qquad $\sigma=[3,1,5,2,6,4]$

        \vspace{0.3cm}

        \begin{tikzpicture}[
            vertex/.style={
                circle,
                draw,
                thick,
                minimum size=8mm,
                inner sep=0pt
            },
            every path/.style={thick}
        ]

        \node[vertex] (1) at (0,0) {1};
        \node[vertex] (3) at (1.3,0) {3};
        \node[vertex] (2) at (2.6,0) {2};
        \node[vertex] (5) at (3.9,0) {5};
        \node[vertex] (4) at (5.2,0) {4};
        \node[vertex] (6) at (6.5,0) {6};

        \draw (1) -- (3);
        \draw (3) -- (2);
        \draw (2) -- (5);
        \draw (5) -- (4);
        \draw (4) -- (6);

        \end{tikzpicture}
    \end{minipage}

    \caption{
   \textbf{ Inversion graphs for two permutations in $S_6$.}
    (a) The permutation $\pi$ has non-trivial cuts and its inversion graph is disconnected. (b) In contrast, $\sigma$ has no non-trivial cuts, and its inversion graph is connected.
    }
    \label{fig:inversion-graphs}
\end{figure}

Now consider a signed permutation $w\in B_{2n}\cap \SO(2n)$, with underlying permutation $\pi$. We can always write $w=D(\vec{\delta})\pi$, where $\vec{\delta}\in\mathbb{F}_2^{2n}$ is a binary vector and $D(\vec{\delta})= \operatorname{diag}\left((-1)^{\delta_1},\ldots,(-1)^{\delta_{2n}}\right)$ a diagonal matrix. For example, let  $w = [3, -1, 2, 4]$. Then $\vec{\delta}=(1,0,0,0)^T$ because the value ending on the first position has a minus sign, while the rest do not. With this, we prove the following lemma.

\begin{lemma}\label{lem:sign_vectors}
Consider a word over the generators  $r_a^{\pm1}$, with underlying permutation $\sigma\in S_{2n}$ and containing
exactly $\inv(\sigma)$ generators. Moreover, assume $\sigma$ is cut-free. Then, any sign vector  $\vec{\delta}$ satisfying
\begin{equation}
    \sum_{i=1}^{2n} \delta_i \equiv \inv(\sigma)
    \pmod 2\,,
\end{equation}
can be realized by changing the orientation of the generators.
\end{lemma}

\begin{proof}

Let us denote by $\vec{\delta}^\circ$ the sign vector associated to the fixed word over $r_a^{\pm1}$. Since the generators have determinant $+1$, Eq.~\eqref{eq:determinant} implies
\begin{equation} \label{eq-ap:parity}
    \sum_{i=1}^m \delta_i^\circ \equiv \inv(\sigma)  \pmod 2\,.
\end{equation}
Choosing between $r_a$ and $r_a^{-1}$ for a given generator in the word amounts to changing both signs of the exchanged Majorana operators (see Eq.~\eqref{eq:transpositions}). 
In terms of the sign vector,
the relative change is
\begin{equation}
    \vec{e}_{u} + \vec{e}_{v},
\end{equation}
where $u,v$ are the values exchanged by $r_a^{\pm1}$ and $\vec{e}_{u}, \vec{e}_{v}$ are canonical basis vectors of $\mathbb F_2^{2n}$. 
The pairs of Majoranas that are swapped in the word are precisely those corresponding to the inversion pairs of $\sigma$. Hence, every edge
$(u,v)$ of the inversion graph $\Gamma_\sigma$ corresponds to one generator in the word. 

Now consider two arbitrary values $u$ and $v$ which are not contiguous. Since $\sigma$ is
cut-free, Lemma~\ref{le:connected_graph} implies that
$\Gamma_\sigma$ is connected. Therefore, there exists a path from $u$ to $v$ in $\Gamma_\sigma$, given by the sequence of edges
\begin{equation}
    (u,v_1),\, (v_1,v_2),\, \ldots,\, (v_{k-1},v)\,.
\end{equation}
Every edge on this path corresponds to one generator in the word. By changing the orientation of all these generators, the resulting relative change in the sign
vector is
\begin{align}
    (\vec{e}_{u} + \vec{e}_{v_1}) + (\vec{e}_{v_1} + \vec{e}_{v_2}) +\cdots + (\vec{e}_{v_{k-1}} + \vec{e}_{v}) = \vec{e}_u + \vec{e}_v\,,
\end{align}
because every intermediate term appears twice and therefore cancels out modulo $2$. 
Thus, we can flip the final signs of any chosen pair of Majoranas $\gamma_u,\gamma_v$ without changing either the underlying permutation or the length of the
word.

Consequently, suppose $\vec{\delta}$ is the target sign vector we wish to implement. Since it must satisfy Eq.~\eqref{eq-ap:parity}, it must differ from $\vec{\delta}^\circ$ in an even number of entries. We can then pair the values whose
signs differ, and realize each pair flip using a path in $\Gamma_\sigma$ as above, without changing the length of the word.

\end{proof}

This lemma implies that, for a cut-free permutation $\sigma$, all sign vectors can be realized using $\inv(\sigma)$ generators, subject to the parity constraint imposed by $\inv(\sigma)$. However, the blocks $w_k$ with determinant $-1$ cannot be completely synthesized using only generators that act within the block, since the generators have determinant $+1$. In particular, given a word with underlying permutation $\pi$ and length $\inv(\pi)$, we need to correct the parity of its sign vector in all blocks with determinant $-1$, using generators that act across contiguous blocks. 

Let
\begin{equation}
    0<c_1<\cdots<c_{t-1}<c_t=2n
\end{equation}
be the cut points of a target $w\in B_{2n}\cap \SO(2n)$, where we have included the trivial endpoint
$c_t=2n$. These define the cut-free blocks
\begin{equation}
    B_j=\{c_{j-1}+1,\ldots,c_j\}\,,
    \qquad j=1,\ldots,t\,,
\end{equation}
with $c_0=0$. 
For convenience, we extend the definition of the cut determinant by
$\Delta_0(w)=1$. Notice also that $\Delta_{2n}(w)=\det(w)=1$.

For each $j\in [t-1]$, we define
\begin{equation}
    b_j =
    \begin{cases}
        1\,, & \Delta_{c_j}(w)=-1\,,\\
        0\,, & \Delta_{c_j}(w)=+1\,,
    \end{cases}
\end{equation}
and set $b_0=0$. Thus, the number of bad cuts is
\begin{equation}
    \nu(w)=\sum_{j=1}^{t-1} b_j\,.
\end{equation}
Since $w$ is block diagonal, we have
\begin{equation}
    \Delta_{c_j}(w)  =
    \prod_{k=1}^j \det\left(w_k\right)\,.
\end{equation}
Consequently,
\begin{equation} \label{eq-ap:detwj}
    \det\left(w_j\right) =
    \Delta_{c_{j-1}}(w)\Delta_{c_j}(w) =
    (-1)^{b_{j-1}+b_j}\,.
\end{equation}

To correct the signs of the block determinants, we use $r_c^2$ at every bad cut, i.e.,
\begin{equation}
    D :=  \prod_{j=1}^{t-1}
    \left(r_{c_j}^2\right)^{b_j}\,.
\end{equation}
Recall that $r_c^2$ leaves the underlying permutation unchanged and
flips the signs at $c$ and $c+1$. Therefore, on the block $B_j$,
the operator $D$ flips one sign if the cut $j-1$ is bad and another
sign if the cut $j$ is bad. Hence,
\begin{equation} \label{eq-ap:detD|Bj}
    \det\left(D[B_j;B_j]\right)  =  (-1)^{b_{j-1}+b_j}\,.
\end{equation}
Consider now
\begin{equation}
    u:= Dw\,.
\end{equation}
The underlying permutation of $u$ is still $\pi=|\tilde w|$, while for every
block,
\begin{equation}
    \det\left(u[B_j;B_j]\right)  = \det\left(w_j\right)
    \det\left(D[B_j;B_j]\right) =1\,,
\end{equation}
where we used Eqs.~\eqref{eq-ap:detwj} and~\eqref{eq-ap:detD|Bj}.

Each  block $u[B_j;B_j]$ is cut-free, and by Lemma~\ref{lem:sign_vectors}  can therefore be realized by
choosing the orientations of a word with underlying permutation $\pi|_{B_j}$ and length $\inv(\pi|_{B_j})$, where $\pi|_{B_j}$ is the restriction of $\pi$ to the set $B_j$. Since
there are no inversions between different cut-free blocks, concatenating
these words gives a word for $u$ of length
\begin{equation}
    \sum_{j=1}^t \inv(\pi|_{B_j})  = \inv(\pi)\,.
\end{equation}

Finally, using that $D^2=\id$, it follows that
\begin{equation}
    w=D u\,.
\end{equation}
There is one factor $r_c^2$ in $D$ for every bad cut, and each such
factor requires two generators from $\RC$. Therefore, this
construction realizes $w$ using
\begin{equation}
    \inv(\pi)+2\nu(w)
\end{equation}
generators, which proves
\begin{equation}
    \ell_{\mathcal{R}}(w) \leq \inv(\pi)+2\nu(w)\,.
\end{equation}
Combining this with the lower bound  established in Lemma~\ref{lem:potential_lower_bound}, we conclude that
\begin{equation}
    \ell_{\mathcal{R}}(w) = \inv(\pi)+2\nu(w) = \Phi(w)\,.
\end{equation}

The construction can moreover be implemented in $\OC\left(n^2\right)$ time on a classical computer. Indeed, the cut points  can be found in $\OC\left(n\right)$ time by scanning $\pi$. Likewise the bad cuts in $w$ can be identified from the parities and signs of the cut-free blocks in linear time. Sorting the underlying permutation of each block using the standard bubble-sort algorithm takes $\OC\left(n^2\right)$ time. Then, to determine the orientations of the braids we need the vector $\vec{\delta}$, which can be found in linear time from $\vec{\delta^\circ}$. Finally, the required changes in orientation needed to implement $\vec{\delta}$ can be done employing the procedure explained in Lemma~\ref{lem:sign_vectors}, by finding a spanning tree of the inversion graph and traversing it, which also takes linear time. Summing over all blocks, the total runtime is $\OC\left(n^2\right)$. 
And since there exist targets with $\ell_{\mathcal R}(w)=\Theta\left(n^2\right)$, this
worst-case runtime is asymptotically tight.

\end{proof}

\section{Details on numerical benchmarks}\label{ap:numerics_path}
In this section, we provide the details of the numerical benchmarks shown in Figure~\ref{fig:benchmark}.

We generate the targets as random walks of length $10n$ over the generating set $\mathcal R$ of Eq.~\eqref{eq:generators_path_garph}. At each step, with probability $\tfrac12$ we apply $S_q$ or $S_q^\dagger$, each with probability $\tfrac12$, to a uniformly random qubit $q$, which is the generator $r_{2q-1}^{\pm1}$; otherwise we apply $R_{xx}^{q,q+1}(\pm\pi/2)$, with a uniformly random sign, to a uniformly random adjacent pair, which is $r_{2q}^{\pm1}$. This ensemble is not uniform on $B_{2n}\cap\SO(2n)$, and the typical optimal gate count depends on it.

We use $n=10,20,30,40,60$ and $80$ qubits ($2n$ Majorana modes) with $10,10,8,8,6$ and $5$ targets, respectively, and the $s$-th target at each size is generated with seed $3000+s$. We use fewer targets at larger sizes. In the sampled targets, the relative range, $(\max-\min)/\mathrm{mean}$, decreases from approximately $83\%$ at $n=10$ to $15\%$ at $n=80$ for our total gate count, and from $96\%$ to $12\%$ for our $R_{xx}$ count. Each walk is stored as a circuit over $\{S,S^\dagger,R_{xx}(\pm\pi/2)\}$, and we read off its signed permutation $w$ from the action of the circuit on the Majorana operators, as dictated by Eq.~\eqref{eq:conjugation_majorana}. All compilers receive the same targets.

For our compiler we use a greedy descent on the potential $\Phi$ of Theorem~\ref{thm:path-optimal}, which returns words of  optimal length. Starting from $w$, it repeatedly applies a generator of $\mathcal R$ that lowers $\Phi$ by one  until it reaches the identity. When several such generators are available, we select the first in the fixed order $r_1,r_1^{-1},r_2,r_2^{-1},\ldots,r_{2n-1},r_{2n-1}^{-1}$, without any preference based on whether the corresponding gate is one- or two-qubit. Because $\ell_R(w)=\Phi(w)$, such a generator exists for every $w\neq e$ (undo the last letter of an optimal word for $w$), so the loop ends after exactly $\Phi(w)=\ell_R(w)$ steps. Thus, the compiler optimizes only the total gate count. The number of $R_{xx}$ gates obtained upper-bounds $\#_2(w)$, whose exact evaluation requires $\mu(\pi)$.

We compare against four methods: two compilers that act directly on signed permutations, the sort-then-fix heuristic and Givens elimination, and two production transpilers, Qiskit~\cite{javadi2024quantum,qiskit} and pytket~\cite{sivarajah2020tket}. The first two sort the permutation first and repair the signs afterwards, without using the choice between $r_a$ and $r_a^{-1}$ to place the signs during the sort. Sort-then-fix sorts $\pi$ by bubble sort, repeatedly swapping the leftmost adjacent pair that is out of order, always with the same one of $r_a$ and $r_a^{-1}$. It then pairs up the entries left with a negative sign, from left to right, and flips each pair $p<q$ back with $r_p^2 r_{p+1}^2\cdots r_{q-1}^2$, which changes only the signs at $p$ and $q$. The Givens-elimination compiler, the method used to compile measurement settings for fermionic classical shadows~\cite{zhao2021fermionic}, goes through the matrix representation $w$ one column at a time from left to right and rotates the nonzero entry of each column up to the diagonal with adjacent rotations, again always of the same orientation. This leaves a diagonal matrix of signs, which it clears in pairs in the same way.

The production transpilers receive the random-walk circuit itself. With Qiskit~2.1~\cite{javadi2024quantum,qiskit}, we run \texttt{transpile} with basis gates $\{S,S^\dagger,R_{xx}\}$ and a line coupling map at optimization levels 1 to 3 and keep the best verified output. We include the input circuit and all verified transpilation outputs as candidates, and select the candidate minimizing total native gate count, with $R_{xx}$ count as a secondary criterion. Failed transpilation attempts and outputs that do not pass verification are discarded. With pytket~2.9~\cite{sivarajah2020tket}, we apply \texttt{CommuteThroughMultis} followed by \texttt{RemoveRedundancies}. The first command commutes single-qubit gates through two-qubit gates, while the second one applies cancellations when possible. 

For verifying our compiler and the two transpilers, we check that the circuit contains only $S$, $S^\dagger$ and $R_{xx}$ gates and that its Clifford tableau equals that of the target, which fixes the unitary up to a global phase. The word returned by our compiler is first converted back into such a circuit. The outputs of sort-then-fix and Givens elimination are words over $\mathcal R$, and we check that they multiply out to $w$. Each $R_{xx}(\pm\pi/2)$ counts as one two-qubit gate. The pytket passes merge some pairs of adjacent $R_{xx}(\pi/2)$ gates into $R_{xx}(\pi)$, and the value we plot is pytket's own two-qubit count, in which such a gate counts twice.

In Fig,~\ref{fig:benchmark}, the plotted points show the mean gate count over the targets at each size, with error bars indicating one sample standard deviation above and below the mean. The lines are unweighted linear least-squares fits to these means, and the multiplicative annotations denote  $a_{\mathrm{baseline}}/a_{\mathrm{ours}}$, comparing the fit growth rates across the sampled range. For total gate count, these ratios are approximately $2.57$, $2.28$, $1.58$ and $1.53$ for Qiskit, pytket, sort-then-fix and Givens elimination, respectively. For $R_{xx}$ count, the corresponding ratios are $2.99$, $2.44$, $1.58$ and $1.50$.

The scripts that generate the targets and run all five compilers are publicly available online~\footnote{\url{https://github.com/yuxuanzhang1995/matchgate-clifford-synthesis}}.

\section{Proof of Theorem~\ref{th:optimal_two-qubit}}\label{ap:proof_optimal_two_qubit}

In this section, we prove the main results related to two-qubit gate count for path graph generators. We restate the results for clarity. 
\begin{theorem}\label{th:optimal_two-qubit_app}
For every $w\in B_{2n} \cap \SO(2n)$ with underlying permutation $\pi=|w|$, the minimum number of two-qubit gates over all words realizing $w$,  $\#_2(w) $,  is given by 
\begin{equation}\label{eq:two_qubit_count_app}
    \#_2(w) = \mu(\pi) + 2\nu_{\operatorname{even}}(w)\,, 
\end{equation}
where $\mu(\pi)$ is the minimum number of even-indexed generators appearing in a reduced word for $\pi$, and $ \nu_{\operatorname{even}}(w)$ is defined as in Eq.~\eqref{eq:nu_even}. 
The quantity $\mu(\pi)$ satisfies
\begin{equation}\label{eq:bounds_mu(pi)_app}
    F(\pi)
    \leq
    \mu(\pi)
    \leq
    2F(\pi)\,, \quad F(\pi) \leq \left\lfloor\frac{n^2}{2}\right\rfloor\,,
\end{equation}
where $F(\pi)$ is the net right (or left) flow of Majoranas over all physical boundaries. The upper bound is constructive and therefore gives a factor-$2$ approximation to $\mu(\pi)$.

Moreover, there exists a word that simultaneously achieves the minimum total length
\begin{equation}
    \ell_R(w)
    =
    \inv(\pi)+2\nu(w),
\end{equation}
as given by Theorem~\ref{thm:path-optimal}, and the minimum two-qubit gate count $\#_2(w)$.

\end{theorem}

\begin{proof}
We first prove Eq.~\eqref{eq:two_qubit_count_app}. For it, we prove lower and upper bounds separately. For the lower bound, consider any word $W$ over $\{r_a^{\pm1}\}$ that realizes $w$, and let
$E(W)$ denote the number of even-indexed generators in $W$. 

Ignoring the signs for a moment, the word $W$ projects to a word
$s_{a_L}\cdots s_{a_i}$ in adjacent transpositions that realizes $\pi$, where both $r_i$ and $r_i^{-1}$ project to the same adjacent transposition $s_i$. In other words, $s_i$ describes the permutation produced by $r_i^{\pm1}$ ignoring the signs.

If the word  $s_{a_L}\cdots s_{a_1}$ is not reduced, the Coxeter deletion property~\cite{bjorner2005combinatorics} allows us to successively delete pairs of letters until we obtain a reduced expression $V$ for the same permutation.
Since by definition every reduced word for $\pi$ contains at least $\mu(\pi)$ even-indexed transpositions, $W$ must contain at least $\mu(\pi)$ generators. Consequently 
\begin{equation}
    |W|\geq  E(W) \geq E(V)\geq \mu(\pi) .
\end{equation}

We now consider an even cut point $c$ of $\pi$ with $\Delta_c(w)=-1$, and consider how many occurrences of $r_a^{\pm1}$ are required for it to happen. The generator  $r_c^{\pm1}$ is the only one moving values across the boundary of the bad cut. Hence, if no $r_c^{\pm1}$ appears in $W$, $\Delta_c$ is necessarily  $+1$ as the two sides of the boundaries never interact.  If $r_c^{\pm1}$ appears only once, $c$ cannot be a cut as it mixes values $c$ and $c+1$ and $\pi([c])= [c]$ would not be satisfied. Hence, a bad cut at $c$ requires at least two occurrences of $r_c^{\pm1}$.

Moreover, because $c$ is a cut point of $\pi$, no reduced word for $\pi$ contains the transposition $s_c$. Indeed, $s_c$ exchanges the two values $c$ and $c+1$ and since $\pi([c])= [c]$, these two values do not form an inversion in $\pi$ and consequently do not cross in a reduced word. 

For even $c$, $s_c$ correspond to two-qubit gates. Thus, every bad even cut contributes at least two additional even-indexed generators beyond those contained in $V$. Summing over all bad even cuts,
\begin{equation}
    E(W)
    \geq
    E(V)+2\nu_{\operatorname{even}}(w).
\end{equation}
Since $V$ is reduced, $E(V)\geq\mu(\pi)$, and therefore
\begin{equation}
    E(W)
    \geq
    \mu(\pi)+2\nu_{\operatorname{even}}(w).
\end{equation}

Since this holds for every word realizing $w$,
\begin{equation}\label{eq:twoq_lower}
    \#_2(w)
    \geq
    \mu(\pi)
    +
    2\nu_{\operatorname{even}}(w).
\end{equation}

For the upper bound we construct a word $W$ for $w\in B_{2n}\cap \SO(2n)$ that uses only $\mu(\pi) + 2\nu_{\operatorname{even}}(w)$ two-qubit gates. The first step is to build a word  for $\pi$ that results in a signed permutation $u'$ such that $|u'| = |w|= \pi$. Among those words, choose one that uses only $\mu(\pi)$ two-qubit gates. Initially, for each transposition $s_{a_i} $ in the reduced word choose $s_{a_i} \mapsto r_{a_i}$. 

Since $u'$ is the signed permutation produced by a reduced word, it does not use any $r_c$ at a cut of $\pi$. This means that $\Delta_c (u') = +1$ for all cuts. However, it is possible that $w$ has some bad cuts. We apply $z_c:= r_c^2$ to $u'$ for all $c$ corresponding to a bad cut of $w$, which results in a permutation $u$ with 
\begin{equation}
    \Delta_c(u) = \Delta_c(w) 
\end{equation}
for all cuts in $w$. This costs exactly $2\nu_{\operatorname{even}}(w)$ two-qubit gates.

We now use this lemma: 
\begin{lemma}
    Consider two signed permutations $u, w\in B_{2n}\cap \SO(2n)$ such that $|\tilde u| = |\tilde w| = \pi$ and $\Delta_c(u) = \Delta_c(w)$ for all cut points $c$ of $\pi$. The sign mismatch between $u$ and $w$ has even parity on every cut-free block of $\pi$. 
\end{lemma}
\begin{proof}
    Consider the sign mismatch vector with components 
    \begin{equation}
        f_i = \begin{cases}
            1 ,\quad  \text{sign of value $i$ differs in $u$ and $w$}
            \\
            0, \quad \text{otherwise}. 
        \end{cases}
    \end{equation}
    We recall the definition of $\Delta_c(u)$: 
    \begin{equation}
        \Delta_c(u)  = \det (u[1\dots,c;1\dots,c]) = \sgn(\pi|_{[c]}) \prod_{i = 1}^c \sgn(\tilde u_i). 
    \end{equation}
    Since $u$ and $w$ have the same underlying permutation, the term $\sgn(\pi|_{[c]})$ is identical for both quantities $\Delta_c(u)$ and $\Delta_c(w)$. Moreover, because \(c\) is a cut point, the first \(c\) positions contain exactly the values in \([c]\). Therefore, 
    \begin{equation}
        \frac{\Delta_c(u)}{\Delta_c(w)}= (-1)^{\sum_{i = 1}^c f_i}.
    \end{equation}
     Since  $\Delta_c (u) = \Delta_c(w)$, we find that 
    \begin{equation}
        \sum_{i = 1}^c f_i \equiv 0 \bmod2.
    \end{equation}
    Next, consider that $\pi$ decomposes into $t$ cut-free blocks $B_1, \dots B_t$. We write the $j$-th block as
    \begin{equation}
        B_j = \{c_{j-1}+1, \dots c_j\},
    \end{equation}
    being $c_1, \dots c_{t-1}$ the cut points and setting $c_0 = 0$ and $c_t = 2n$. Since the two ends of the $j$-th block are cut points, we have 
    \begin{equation}
         \sum_{i = 1}^{c_j} f_i  =  \sum_{i = 1}^{c_{j-1}} f_i \equiv 0 \bmod2.
    \end{equation}
    Therefore, the number of sign defects inside each block satisfies 
    \begin{equation}
         \sum_{i \in B_j} f_i = \sum_{i = 1}^{c_j} f_i - \sum_{i = 1}^{c_{j-1}} f_i\equiv 0 \bmod2.
    \end{equation}
    This shows that for every block $B_j$, the sign mismatch has even parity, which proves the lemma. 
\end{proof}

To correct the even signs inside each block, we use the fact that the inversion graph of every cut-free block is connected (Lemma~\ref{le:connected_graph}). If we want to flip the signs of two Majoranas in the same block, we choose a path that connects them and apply the change $r_a\leftrightarrow r_a^{-1}$ to every crossing corresponding to an edge on the path. This effectively changes the signs of the two endpoints Majoranas. Importantly, this operation does not add any extra gate count and it is enough to correct all signs by pairing the remaining defects.

To sum up, the construction uses $2\nu_{\operatorname{even}}$ two-qubit gates to correct even bad cuts and $\mu(\pi)$ two-qubit gates  in the reduced word.  All remaining sign corrections do not introduce additional two-qubit gates. This proves that 
\begin{equation}\label{eq:twoq_upper}
    \#_2(w)\leq\mu(\pi)+ 2\nu_{\operatorname{even}}(w).
\end{equation}

Combining Eqs.~\eqref{eq:twoq_lower} and \eqref{eq:twoq_upper} gives
\begin{equation}
    \#_2(w) = \mu(\pi)+ 2\nu_{\operatorname{even}}(w).
\end{equation}

We now prove the simultaneous optimality in two-qubit gates and optimal length. Consider the construction used in the upper bound above. We start from a reduced word
for $\pi$ that contains exactly $\mu(\pi)$ even-indexed
transpositions. Since the word is reduced, its total length is
$\inv(\pi)$.

For every bad cut $c$ of $w$, we introduce the correction $z_c=r_c^2$. There are $\nu(w)$ bad cuts, and each correction adds exactly two generators. Hence, these corrections increase the total length by $2\nu(w)$. Among them, only the corrections associated with even cuts contribute to the two-qubit gate count. Therefore, they add exactly $ 2\nu_{\operatorname{even}}(w)$ two-qubit gates.

As shown in the proof of Theorem~\ref{th:optimal_two-qubit_app}, after these corrections the remaining sign mismatch has even parity within every cut-free block. It can therefore be corrected by changing
the orientations
\[
    r_a\leftrightarrow r_a^{-1}
\]
along the connected inversion graph of each block. These orientation changes do not modify either the total number of gates or the number
of two-qubit gates.

Consequently, the resulting word $W$ realizes $w$ and satisfies
\begin{equation}
    |W| 
    =
    \inv(\pi)+2\nu(w) = \ell_R(w),
\end{equation}
and simultaneously has a number of even-indexed generators 
\begin{equation}
    E(W)
    =
    \mu(\pi)+2\nu_{\operatorname{even}}(w)  =\#_2 .
\end{equation}

By Theorem~\ref{thm:path-optimal}, the first quantity is the minimum possible total length, and as proven above, the second is the minimum possible number of two-qubit gates. We conclude that same word $W$ simultaneously achieves both minima.

The exact evaluation of the routing term $\mu(\pi)$ and how to construct a word that achieves it remains open. However, we can derive the bounds given in Eq.~\eqref{eq:bounds_mu(pi)_app}. To do so, we use the following lemma:

\begin{lemma}
    For every $\pi\in S_{2n}$, the minimum number of even-indexed generators required for a reduced word for $\pi$ is bounded by
\begin{equation}
    2F(\pi)-\xi(\pi)
    \leq
    \mu(\pi)
    \leq
    2F(\pi),
\end{equation}
where $F(\pi)$ is the net right (or left) flow of Majoranas over all boundaries and $\xi(\pi)$ is the pairing number. 
Both bounds are computable in polynomial time, and the upper bound is constructive.
\end{lemma}

\begin{proof}
For the upper bound,  we first construct a sequence of generators realizing $\pi$ using at most $2F(\pi)$ even-indexed generators. To ease the proof, we track the inverse permutation $\rho = \pi ^{-1}$. If the current permutation $\sigma$ is updated by left multiplication, $\sigma \rightarrow s_a \sigma$,  then its inverse transforms as $\rho \rightarrow \rho s_a$. Thus, in the window notation of $\rho$, the generator $r_a$ exchanges the entries at positions $a$ and $a+1$. 

Define physical modes as bins containing two Majoranas as
\begin{equation}
    \{1,2\}, \{3,4\}\dots \{2n-1,2n\}
\end{equation}
We consider each mode as a bin, and label every Majorana $i$ by its target mode $t_i$, according to the permutation $w$. Odd-indexed generators only exchange Majoranas within the same bin, whereas even-indexed generators exchange Majoranas across two neighboring bins.

A Majorana is said to be a right-moving if its target mode (dictated by $\rho$) lies in the right, in other words, its source mode $s_i = \left\lceil i/2\right\rceil$ is smaller than its target mode $t_i = \left\lceil\pi(i)/2\right\rceil$. Equivalently, a left-moving Majorana satisfies $s_i>t_i$.

Consider a boundary $\beta$ between modes $\beta$ and $\beta+1$. We consider two types of exchange across this boundary: 
\begin{itemize}
    \item Productive exchange: if it moves a Majorana whose target lies to the right of $\beta$ from left to right and, simultaneously, moves a Majorana whose target lies to the left of $\beta$ from right to left. Thus, an exchange is productive if both Majoranas move towards their target modes.
    \item Neutral exchange: if a right-moving Majorana needs to get through the next bin and it forces to move the Majoranas from bin $\beta+1$ (who actually belonged there) to bin $\beta$. In other words, we are temporarily displacing the Majoranas in bin $\beta+1$.
\end{itemize}
We use the following routing procedure. Whenever a productive exchange is available, we perform one. If no productive exchange is available, we choose a right-moving Majorana $u$ whose current mode $b$ is as far right as possible, and exchange it with one of the Majoranas in mode $b+1$.

We first show that this second type of exchange is always possible and is a neutral exchange. Since $u$ is right-moving,
\begin{equation}
    t_u>b.
\end{equation}
Because no productive exchange is available, both Majoranas currently in mode $b+1$ must have target mode at least $b+1$. Otherwise, one of them would need to move to the left and could be productively exchanged with $u$.

Moreover, neither Majorana in mode $b+1$ can have target mode strictly larger than $b+1$. If this happened, that Majorana would itself be a right-mover located in a mode to the right of $b$, contradicting the maximality of $b$ used in the routing procedure. Hence both Majoranas in mode $b+1$ have target mode exactly $b+1$.

Since there are exactly two Majoranas whose target is mode $b+1$, we must have
\begin{equation}
    t_u>b+1.
\end{equation}
Therefore, exchanging $u$ with one of the Majoranas in mode $b+1$ moves $u$ one mode towards its target, while temporarily displacing a Majorana from its target mode. Therefore, this corresponds to a neutral exchange. 

We now bound the number of productive and neutral exchanges. For each boundary $\beta$, define
\begin{equation}
    q_\beta
    :=
    \#\left\{
        i:
        \text{$i$ is currently in a mode at most $\beta$ and $t_i>\beta$}
    \right\}.
\end{equation}
Initially, $  q_\beta=F_\beta$. We recall that 
\begin{equation}
    F_\beta
    :=
    \#\{i:s_i\leq\beta<t_i\}
\end{equation}
is the number of right-moving Majoranas across boundary $\beta$. At the end of the routing $ q_\beta=0$, as we have moved all the modes in their target position. 

A productive exchange across boundary $\beta$ decreases $q_\beta$ by one: a Majorana whose target lies to the right leaves the left side, while the Majorana entering the left side has its target there. Therefore,
\begin{equation}
    q_\beta\longmapsto q_\beta-1.
\end{equation}

By contrast, a neutral exchange leaves $q_\beta$ unchanged. A
right-moving Majorana leaves the left side, decreasing $q_\beta$ by one, but a Majorana whose target is mode $\beta+1$ is simultaneously
displaced to the left side, increasing $q_\beta$ by one. Hence,
\begin{equation}
    q_\beta\longmapsto q_\beta.
\end{equation}
Exchanges across other boundaries do not change $q_\beta$.

It follows that exactly $F_\beta$ productive exchanges must occur across boundary $\beta$. Denote this number by $P_\beta$:
\begin{equation}
    P_\beta=F_\beta.
\end{equation}

Let $N_\beta$ denote the number of neutral exchanges across the same boundary. Consider one such exchange. It displaces a Majorana $v$ whose target mode is $\beta+1$ from mode $\beta+1$ into mode $\beta$. This Majorana cannot subsequently move farther to the left under the routing rule. Since it must eventually return to its target mode, it must later cross boundary $\beta$ from left to right.

After this exchange, $v$ is a right-moving Majorana. We claim that $v$ cannot subsequently move farther to the left. A productive exchange cannot do so, since this would move $v$ away from its target mode.
Moreover, $v$ cannot be displaced to the left by another neutral exchange. Indeed, neutral exchanges are performed only when no
productive exchange is available, in which case the algorithm chooses a right-moving Majorana whose current mode is as far right as possible. Since $v$ is already right-moving in mode $\beta$, a Majorana in mode
$\beta-1$ cannot be selected to displace it further to the left, as this would contradict the routing rule. 

Consequently, we associate each neutral exchange with the later productive exchange that returns its displaced Majorana to mode $\beta+1$. Different neutral exchanges must be associated with different productive exchanges, because a single productive swap moves only one Majorana from left to right and therefore can return at most one displaced Majorana. Hence the association is injective, and
\begin{equation}
    N_\beta\leq P_\beta=F_\beta.
\end{equation}

Taking this into account, the total number of inter-mode exchanges across boundary
$\beta$ is at most
\begin{equation}
    P_\beta+N_\beta
    \leq
    2F_\beta.
\end{equation}
Summing over all boundaries gives
\begin{equation}
    \sum_{\beta=1}^{n-1}(P_\beta+N_\beta)
    \leq
    2\sum_{\beta=1}^{n-1}F_\beta
    =
    2F(\pi).
\end{equation}
This routing gives a left-multiplication synthesis of $\pi$, suppose that the routing of $\rho=\pi^{-1}$ uses the generators $s_{a_1},\ldots,s_{a_L}$ in this order. Since the routing acts by right multiplication, $\rho=s_{a_1}\cdots s_{a_L}$. Taking inverses and using $s_a^{-1}=s_a$, we obtain $\pi=\rho^{-1}=s_{a_L}\cdots s_{a_1}$.
Therefore, applying the same generators in the temporal order $s_{a_1},\ldots,s_{a_L}$ by left multiplication produces $\pi$. 

The resulting word does not need to be reduced. However, by the Coxeter deletion property, we may delete letters until obtaining a reduced word for the same permutation, without increasing the number of even-indexed generators. Therefore,
\begin{equation}
    \mu(\pi)\leq2F(\pi).
\end{equation}

We now prove the lower bound 
\begin{equation}
    \mu(\pi)\geq2F(\pi)-\xi(\pi).
\end{equation}
Choose a reduced word $V$ for $\pi$ that achieves $\mu(\pi)$. As in the upper-bound argument, we track the inverse permutation $\rho=\pi^{-1}$. The same temporal sequence of generators defining $V$ gives a reduced positional routing of $\rho$. For each mode boundary $\beta$, let $x_\beta$ denote the number of swaps across that boundary in this routing. Since every even-indexed generator crosses exactly one mode boundary,
\begin{equation}
    \mu(\pi)
    =
    \sum_{\beta=1}^{n-1}x_\beta.
\end{equation}

Consider a boundary $\beta$. There are $F_\beta$ Majoranas that start to the left of $\beta$ and end to its right. For each such Majorana, consider its last crossing of $\beta$. Since the Majorana finishes on the right, this last crossing must be from left to right. We call these the $R$-essential crossings. Thus, there are exactly $F_\beta$ $R$-essential crossings.

By conservation, there are also $F_\beta$ Majoranas that start to the right of $\beta$ and end to its left. Their last crossings are necessarily from right to left. We call these the $L$-essential crossings. Again, there are exactly $F_\beta$ of them.

These essential crossings consider one unavoidable crossing for every Majorana that must change sides of the boundary. Other swaps across $\beta$ may also occur, for example crossings that are not the last crossing of either Majorana. A single swap across $\beta$ may be both $R$-essential and $L$-essential, realizing one required crossing in each direction simultaneously. 

Let $M_\beta$ denote the number of such doubly essential swaps. The union of the $R$-essential and $L$-essential crossings therefore contains
\begin{equation}
    2F_\beta-M_\beta
\end{equation}
swaps. Since all of these are among the $x_\beta$ swaps performed across boundary $\beta$,
\begin{equation}
    x_\beta
    \geq
    2F_\beta-M_\beta.
\end{equation}

Summing over all boundaries gives
\begin{align}
    \mu(\pi)
    =
    \sum_{\beta=1}^{n-1}x_\beta
    \geq
    2\sum_{\beta=1}^{n-1}F_\beta
    -
    \sum_{\beta=1}^{n-1}M_\beta
    =
    2F(\pi)-M,
\end{align}
where we used
\begin{equation}
    M:=\sum_{\beta=1}^{n-1}M_\beta.
\end{equation}

It remains to upper bound $M$. For it, consider the sets 
\begin{equation}
    \mathcal R = \{i : t_i>s_i\}, \quad \mathcal L =  \{i : t_i<s_i\}.
\end{equation}
We now construct a bipartite graph, with nodes $\mathcal R$ in one partition and $\mathcal L$ in the other. We connect two vertices $(r,l)$ between $\mathcal R$ and $\mathcal L$ with $r\in \mathcal R$ and $l\in \mathcal L$ if
    \begin{equation}
    s_r\leq\beta<t_r,
    \qquad
    t_\ell\leq\beta<s_\ell \quad \text{for some boundary $\beta$.}
\end{equation}
This condition implies that $r$ and $l$ must form an inversion of $\pi$. These edges represent the possible Majoranas that can have a double essential swap. 

Since each Majorana $r$ can participate in up to $|t_r-s_r|$ doubly essential crossings, the vertex $r$ gets capacity at most $|t_r-s_r|$. Similarly, vertex $l$ gets capacity $|s_l-t_l|$. All edges $(r,l)$ have capacity one, as in the corresponding reduced routing of $\rho$, a given pair of Majoranas can cross at most once.

Given this bipartite graph let $\xi(\pi)$ denote the maximum number of edges that can be selected subject to the constraints that each edge is selected at most once and each vertex $i$ is incident to at most $|t_i-s_i|$ selected edges.
Therefore, the doubly essential swaps define a feasible
degree-constrained matching in the grap. By maximality of $\xi(\pi)$,
\begin{equation}
    M\leq\xi(\pi).
\end{equation}
Substituting this into the previous inequality gives
\begin{equation}
    \mu(\pi)
    \geq
    2F(\pi)-\xi(\pi).
\end{equation}
Finally, combining the upper and lower bounds gives
\begin{equation}
    2F(\pi)-\xi(\pi)
    \leq
    \mu(\pi)
    \leq
    2F(\pi).
\end{equation}
\end{proof}

Moreover, since the total capacity of the right-moving vertices is
\begin{equation}
    \sum_{r\in \mathcal R}(t_r-s_r)=F(\pi),
\end{equation}
the number of selected edges in any feasible matching is at most $F(\pi)$. Therefore, 
\begin{equation}
    \xi(\pi)\leq F(\pi),
\end{equation}
which results in 
\begin{equation}
    F(\pi)
    \leq
    \mu(\pi)
    \leq
    2F(\pi).
\end{equation}
This bound gives a constructive factor-$2$ approximation to $\mu(\pi)$.

It remains to bound $F(\pi)$  in terms of the number of physical modes $n$. For each boundary
$\beta\in\{1,\ldots,n-1\}$, at most $2\beta$ Majoranas start to the
left of the boundary, while at most $2(n-\beta)$ target positions lie
to its right. Therefore,
\begin{equation}
    F_\beta
    \leq
    2\min\{\beta,n-\beta\}.
\end{equation}
Summing over all physical-mode boundaries gives
\begin{align}
    F(\pi)
    =
    \sum_{\beta=1}^{n-1}F_\beta 
    \leq
    2\sum_{\beta=1}^{n-1}\min\{\beta,n-\beta\} 
    =
    \left\lfloor\frac{n^2}{2}\right\rfloor.
\end{align}

We now show that both bounds can be evaluated in polynomial time. The quantity $F(\pi)$ is obtained directly from the source and target modes of the Majoranas by summing the boundary flows $F_\beta$, and can therefore be computed in polynomial time.

For the upper bound, the routing procedure used in the proof is constructive. At each step, one can determine whether a productive exchange is available or otherwise select the rightmost right-moving Majorana. The proof above shows that the total number of inter-mode exchanges is at most $2F(\pi)$. Since $F(\pi)\leq \lfloor n^2/2\rfloor$, the number of such steps is
polynomial in $n$, and hence the resulting routing can be generated in polynomial time.

For the lower bound, the bipartite graph defining $\xi(\pi)$ has at most $2n$ vertices and at most $\OC(n^2)$ edges. The vertex capacities are given by $|t_i-s_i|$, and $\xi(\pi)$ is the maximum cardinality degree-constrained matching in this graph. This is a standard capacitated bipartite matching problem and can be solved in polynomial time, for example by reducing it to a maximum-flow problem~\cite{schrijver2003combinatorial}.
\end{proof}

\section{Proof of Theorem~\ref{thm:all-to-all}}
\label{ap:proof-all-to-all}

In this section, we prove that optimal synthesis is efficiently achievable with Majorana braids and all-to-all connectivity. 

\begin{theorem}
\label{thm:all-to-all_app}
For every signed permutation $w\in B_{2n}\cap \SO(2n)$ with underlying permutation $\pi=|w|$, the minimum number of generators
from $\RC_{\mathrm{all}}$ required to synthesize $w$ is
\begin{equation}
    \ell_{\mathrm{all}}(w)  = \bigl(2n-\kappa(\pi)\bigr)+d_-(w)\,,
    \label{eq:all-to-all-length_app}
\end{equation}
where $\kappa(\pi)$ is the number of cycles in $\pi$ and $d_-(w)$ is the number of defective cycles. Moreover, $\ell_{\mathrm{all}}(w)\leq2n$, and an optimal word is computable in $\mathcal \OC\left(n^2\right)$ time on a classical computer.
\end{theorem}

\begin{proof}

We recall that the generating set is
\begin{equation}
    \RC_{\mathrm{all}} :=
    \left\{
        r_{ab}\,, r_{ab}^{-1} :\; 1\leq a<b\leq 2n
    \right\}\,.
\end{equation}
For a cycle $C$ of $\pi$, we denote
\begin{equation}
    \delta_w(C)  := \det\left(w[C;C]\right)  \in\{+1,-1\}\,,
\end{equation}
where $w[C;C]$ is the restriction of $w$ to the coordinates in $C$. We say that $C$ is defective if $\delta_w(C)=-1$, and consider the number of defective cycles,
\begin{equation}
    d_-(w) :=
    \#\left\{
        C\text{ cycle of }\pi:
        \delta_w(C)=-1
    \right\}\,.
\end{equation}
We furthermore denote by $\kappa(\pi)$ the number of cycles of $\pi$, including fixed points.

The proof follows the same strategy as in Appendix~\ref{ap:proof_main_thm}. Namely, we define
\begin{equation}
    \Phi_{\mathrm{all}}(w) :=
    \bigl(2n-\kappa(\pi)\bigr)+d_-(w)\,,
    \label{eq:all-to-all-potential}
\end{equation}
and show that this potential characterizes the optimal word length by providing matching lower and upper bounds on $\ell_{\mathrm{all}}(w)$.

We start by proving that $\Phi_{\mathrm{all}}(w)$ gives a lower bound.
\begin{lemma}
\label{lem:all-to-all-potential}
For every $w\in  B_{2n}\cap \SO(2n)$ and every
$g\in \RC_{\mathrm{all}}$,
\begin{equation}
    \Phi_{\mathrm{all}}(gw) - \Phi_{\mathrm{all}}(w) \in\{-1,+1\}\,.
    \label{eq:all-to-all-potential-change}
\end{equation}
\end{lemma}

\begin{proof}
Multiplication by $g=r_{ab}^{\pm1}$ has one of two possible effects on the cycle decomposition of $\pi$: it either merges two distinct cycles or splits a single cycle into two. We analyze both cases separately.

First, consider the case when $a$ and $b$ belong to different cycles $C_1$ and $C_2$. Multiplication by $r_{ab}^{\pm1}$ merges them into a single
cycle $C'=C_1\cup C_2$ of $gw$. Since two cycles merge, we have one less cycle in $gw$.  
Consequently, the difference in the first term of the potential is 
\begin{equation}
    \Delta(2n-\kappa)=+1\,.
\end{equation}

Moreover, since the non-trivial action of $g$ lies entirely inside the coordinates in $C'$, and $\det(g[C';C'])=+1$, we have
\begin{equation}
    \det\left(gw[C';C']\right) = \det\left(w[C';C']\right)\,.
    \label{eq:all-to-all-det-invariance}
\end{equation}
Before applying $g$, the restriction of $w$ to $C'$ decomposes into the two cycle blocks,
\begin{equation}
    w[C';C'] =w[C_1;C_1]\oplus w[C_2;C_2]\,.
\end{equation}
Therefore
\begin{equation}
    \delta_{gw}(C')=\delta_w(C_1)\delta_w(C_2)\,.
\end{equation}
Hence, regarding the second term of the potential, there are three possibilities: 
\begin{itemize}
    \item If $C_1$ and $C_2$ are both non-defective (i.e., $\delta_w(C_1)=\delta_w(C_2)=1$), then $C'$ is non-defective (i.e., $\delta_{gw}(C')=1$), so $\Delta d_-=0$.

    \item If exactly one of $C_1,C_2$ is defective, then $C'$ is defective
    and again $\Delta d_-=0$.

    \item If $C_1$ and $C_2$ are both defective, then $ \delta_{gw}(C)=1$,
    so the two defective cycles are replaced by one non-defective cycle. Therefore, $\Delta d_-=-2$.
\end{itemize}

Consequently, we find
\begin{equation}
    \Delta\Phi_{\mathrm{all}}  =
    \begin{cases}
        -1\,, & \text{if the two cycles are both defective}\,,\\[1mm]
        +1\,, & \text{otherwise}\,.
    \end{cases}
    \label{eq:merge-potential}
\end{equation}

\medskip

We now consider the second case, when $a$ and $b$ belong to the same cycle $C$ in $\pi$. Multiplication by $r_{ab}^{\pm1}$ splits $C$ into two cycles $C_1'$ and $C_2'$. Splitting one cycle into two gives 
\begin{equation}
    \Delta(2n-\kappa)=-1\,.
\end{equation}

 Moreover, the same reasoning as before leads to 
\begin{equation}
    \delta_{gw}(C_1')\delta_{gw}(C_2') = \delta_w(C)\,.
    \label{eq:split-det}
\end{equation}

Again, there are three possibilities for the change in the second term of the potential:

\begin{itemize}
    \item If $C$ is defective, then by Eq.~\eqref{eq:split-det}, exactly one of $C_1',C_2'$ is defective. Therefore,
the number of defective cycles is unchanged, so $\Delta d_-=0$. 

    \item If $C$ is non-defective, and $C_1',C_2'$ are also both non-defective, we have $\Delta d_-=0$. 

    \item If $C$ is non-defective, and $C_1',C_2'$ are  both defective, we have $\Delta d_-=+2$.
\end{itemize}

Therefore, in all possible cases,
\begin{equation}
    \Phi_{\mathrm{all}}(gw) - \Phi_{\mathrm{all}}(w)  \in\{-1,+1\}\,,
\end{equation}
which proves the lemma.
\end{proof}

\medskip

Since $\Phi_{\mathrm{all}}(e)=0$ and each application of a generator changes the potential by at most one in absolute value, any word over $\RC_{\mathrm{all}}$ producing
$w$ from the identity must contain at least $\Phi_{\mathrm{all}}(w)$ generators. Hence,
\begin{equation}
    \ell_{\mathrm{all}}(w)   \geq   \Phi_{\mathrm{all}}(w)\,.
    \label{eq:all-to-all-lower-bound}
\end{equation}

We now prove the converse inequality by providing an algorithm that efficiently synthesizes $w$ using $\Phi_{\mathrm{all}}(w)$ generators. 

Suppose first that the underlying permutation $\pi$ contains a cycle $C$ of length at least two. Choose two distinct elements $a,b\in C$. Multiplication by $r_{ab}^{\pm1}$ splits $C$ into two cycles.

If $C$ is defective, we have shown that either orientation decreases $\Phi_{\mathrm{all}}$ by one. If $C$ is non-defective, choose the orientation that produces two non-defective cycles. As shown in the proof of Lemma~\ref{lem:all-to-all-potential},
this again gives $\Delta\Phi_{\mathrm{all}}=-1$.
Repeating this procedure eventually removes every nontrivial cycle of $\pi$, so that the underlying permutation becomes the identity.

Once we follow these steps, we have $|\tilde w'| = \pi' = e$ and therefore every cycle is a fixed point. We may have defective fixed points, which correspond to negative entries. Since $\det(w')=+1$, the number of negative fixed points is even.
To fix this, choose a pair of defective fixed points $a$ and $b$. On these two coordinates, the signed permutation is $(-a,-b)$. Using the established convention,
\begin{equation}
    (-a,-b)
    \xmapsto{\,r_{ab}\,}
    (b,-a)
    \xmapsto{\,r_{ab}\,}
    (a,b).
    \label{eq:defective-fixed-points}
\end{equation}
The first application merges the two defective fixed points into a
non-defective $2$-cycle, decreasing $\Phi_{\mathrm{all}}$ by one. The
second application splits this non-defective $2$-cycle into two
non-defective fixed points, again decreasing $\Phi_{\mathrm{all}}$ by
one.

Therefore, from every $w\neq e$, there exists a generator that decreases
the potential by exactly one. Iterating this descent gives a word $W_{\mathrm{sort}}$ satisfying $ W_{\mathrm{sort}}w=e$ and $|W_{\mathrm{sort}}|=\Phi_{\mathrm{all}}(w)$.
Consequently, $w=W_{\mathrm{sort}}^{-1}$ and reversing the word and inverting each generator does not change its length. Hence,
\begin{equation}
    \ell_{\mathrm{all}}(w)
    \leq
    \Phi_{\mathrm{all}}(w)
    =
    2n-\kappa(\pi)+d_-(w).
    \label{eq:all-to-all-upper-bound}
\end{equation}

Combining Eqs.~\eqref{eq:all-to-all-lower-bound} and
\eqref{eq:all-to-all-upper-bound}, we conclude that
\begin{equation}
    \ell_{\mathrm{all}}(w)
    =
    2n-\kappa(\pi)+d_-(w)
    .
\end{equation}

The expression for the optimal length also determines the diameter of the Cayley graph associated with $R_{\mathrm{all}}$. Since the generating set is symmetric, the Cayley distance from the identity to $w$ is precisely $\ell_{\mathrm{all}}(w)$. Moreover, for every signed permutation,
\begin{equation}
    d_-(w)\leq \kappa(\pi),
\end{equation}
since every defective cycle is, in particular, a cycle of $\pi$. Therefore,
\begin{equation}
    \ell_{\mathrm{all}}(w)
    =
    2n-\kappa(\pi)+d_-(w)
    \leq 2n.
\end{equation}
This upper bound is attained, for example, by
$w=-e$. In this case $\pi=e$, so $\kappa(\pi)=2n$, while all $2n$ fixed-point cycles are defective and hence $d_-(w)=2n$. Thus
\begin{equation}
    \ell_{\mathrm{all}}(-e)=2n,
\end{equation}
and consequently
\begin{equation}
    \operatorname{diam}
    \left[
        \operatorname{Cay}
        \left(
            B_{2n}\cap\SO(2n),
            \mathcal{R}_{\mathrm{all}}
        \right)
    \right]
    =
    2n.
\end{equation}

The constructive proof directly provides an optimal synthesis algorithm. Starting from $w$, at each step we choose a generator $g\in R_{\mathrm{all}}$ such that $\Phi_{\mathrm{all}}(gw) = \Phi_{\mathrm{all}}(w)-1$, and repeat until the identity is reached. A suitable generator can be found in $\mathcal O(n)$ time. Indeed, from the window notation of $w$, the cycle decomposition of the underlying permutation and the defect of each cycle can be obtained by traversing the $2n$ entries once. If a nontrivial cycle is present, the proof specifies a pair of elements $a,b$ in that cycle and an orientation $r_{ab}^{\pm1}$ that decreases the potential. If the underlying permutation is already the identity, the same traversal identifies a pair of defective fixed points $(a,b)$, which can be removed using $r_{ab}^2$.

Since every generator selected in this way decreases $\Phi_{\mathrm{all}}$ by exactly one, the number of iterations is precisely the optimal length $\ell_{\mathrm{all}}(w)$. Moreover, the Cayley-graph diameter gives
\begin{equation}
    \ell_{\mathrm{all}}(w)
    =
    \Phi_{\mathrm{all}}(w)
    \leq 2n,
\end{equation}
so at most $\mathcal O(n)$ generators have to be selected. Performing $\mathcal O(n)$ work to determine each of at most $\mathcal O(n)$ generators therefore gives a total classical runtime of $\mathcal O(n^2)$.

\end{proof}

\section{Proof of Theorems~\ref{th:tree-npcomplete} and~\ref{th:tree-depth}}
\label{ap:trees}

In this appendix, we prove NP-completeness of optimal matchgate-Clifford synthesis on tree connectivity graphs, namely Theorems~\ref{th:tree-npcomplete} and~\ref{th:tree-depth}.

\begin{theorem} \label{thm:trees_app}
    The decision version of gate-optimal matchgate-Clifford synthesis with Majorana braids on tree connectivity graphs is NP-complete.
\end{theorem}

\begin{proof}

Let $T=(V,E)$ be a tree on the vertex set $V=[2n]$. For every edge $(a,b)\in E$, we denote by $r_{ab}$ and by $r_{ab}^{-1}$ the signed transpositions in Eq.~\eqref{eq:transpositions} that swap $\gamma_a$ and $\gamma_b$.
We write $\ell_T(w)$ for the minimum number of such generators required to synthesize $w$.

First, we show that the problem is NP-hard by reducing from optimal token swapping on a tree. The latter asks for an optimal sequence of swaps along the edges of the tree that synthesizes a target permutation of tokens placed on the vertices, and its decision version is NP-Complete~\cite{aichholzer2021hardness}. Given an input permutation  $\pi\in S_{2n}$, we associate to it a signed permutation $w^*(\pi)\in B_{2n}\cap \SO(2n)$, with the following sign configuration: For every even-length cycle $C$ of $\pi$, we assign a negative sign to the smallest element in $C$, and positive signs to all other elements.  We denote the resulting signs by $s_j^*(\pi)\in\{\pm1\}$, and define the signed permutation $w^*(\pi)$ by
\begin{equation}
    w^*(\pi)\vec{e}_j=s_j^*(\pi)\vec{e}_{\pi(j)}\,.
\end{equation}
Therefore, for every cycle $C$ of $\pi$, the product of the assigned signs is $-1$ if its length $|C|$ is even and $+1$ if it is odd, i.e.,
\begin{equation}
    \prod_{j\in C}s_j^*(\pi)=\operatorname{sgn}(C)\,.
\label{eq-ap:balanced-cycle}
\end{equation}
The same identity holds for every $\pi$-invariant set, since such a set is a union of cycles. In particular, we have that
\begin{equation}
    \prod_{j=1}^{2n} s_j^*(\pi) = \operatorname{sgn}(\pi)\,,
\end{equation}
and so we arrive at
\begin{equation}
\det\bigl(w^*(\pi)\bigr)=\operatorname{sgn}(\pi)\prod_{j=1}^{2n} s_j^*(\pi)=1\,.
\end{equation}
Hence, $w^*(\pi)\in B_{2n}\cap\operatorname{SO}(2n)$, and it is clear that $w^*(\pi)$ can be constructed in polynomial time from the cycle decomposition of $\pi$~\footnote{The cycle decomposition of $\pi$ can also be efficiently computed on a classical computer.}.

We next show that the signs of $w^*(\pi)$ can be efficiently synthesized simply by choosing between $r_{ab}$ and $r_{ab}^{-1}$ for each braid in an arbitrary routing of $\pi$.

\begin{lemma}
\label{lem:orientation}
Let $S=S_{a_L,b_L}\cdots S_{a_1,b_1}$ be any sequence of swaps along edges of $T$ that realizes $\pi$. There always exist orientations $\varepsilon_1,\ldots,\varepsilon_L\in\{\pm1\}$ such that
\begin{equation}
    r_{a_Lb_L}^{\varepsilon_L}\cdots r_{a_1b_1}^{\varepsilon_1}=w^*(\pi)\,,
\label{eq-ap:oriented-routing}
\end{equation}
and they can be found in polynomial time on a classical computer. 
\end{lemma}

\begin{proof}
Let 
\begin{equation}
    \sigma_i=S_{a_i\,b_i}\cdots S_{a_1\,b_1}\,,\qquad \sigma_0=\id\,,
\end{equation}
so that $\sigma_L=\pi$. 
Immediately before step $i$, the two Majoranas exchanged by the swap $S_{a_i,b_i}$ are
\begin{equation} \label{eq-ap:swapped-Majoranas}
    \gamma_{\sigma_{i-1}^{-1}(a_i)}\,,\qquad \gamma_{\sigma_{i-1}^{-1}(b_i)}\,.
\end{equation}
We define the swap multigraph $H(S)$ on the token set $\gamma_1,\ldots,\gamma_{2n}$ by adding an edge $\left(\gamma_{\sigma_{i-1}^{-1}(a_i)},\gamma_{\sigma_{i-1}^{-1}(b_i)}\right)$ for every step $i$. Clearly, the set of Majoranas on a connected component of $H(S)$ is $\pi$-invariant, as $\pi$ acts on this set by transpositions whose endpoints belong to the same component.

Let us consider
\begin{equation} \label{eq-ap:standard-orientation}
    w^\circ =r_{a_Lb_L}\cdots r_{a_1b_1}\,,
\end{equation}
such that  $w^\circ \vec{e}_j= s_j(w^\circ)\vec{e}_{\pi(j)}$. We aim to determine what sign configurations can be attained by changing the orientation of the braids in Eq.~\eqref{eq-ap:standard-orientation}. In particular, we note that we can write an arbitrary signed permutation as $w=w^\circ D\left(\vec{\delta}\right) $, where $\vec{\delta}\in\mathbb{F}_2^{2n}$ and $D(\vec{\delta})= \operatorname{diag}\left((-1)^{\delta_1},\ldots,(-1)^{\delta_{2n}}\right)$. 
Choosing between $r_{a_ib_i}$ and $r_{a_ib_i}^{-1}$ flips the signs of both swapped Majoranas in Eq.~\eqref{eq-ap:swapped-Majoranas}, corresponding to $\vec{\delta} = \vec{e}_{\sigma_{i-1}^{-1}(a_i)} + \vec{e}_{\sigma_{i-1}^{-1}(b_i)}$. Hence, $\vec{\delta}$ is in the span of binary vectors with even Hamming weight on every connected component. Conversely, any binary vector with even Hamming weight on a connected component can be realized by orienting the braids. Indeed, since there is always a path between any two Majoranas (by definition of connectedness), we can implement $\vec{\delta}=\vec{e}_j +\vec{e}_{j'}$, for any $j,j'$ in the connected component (see Lemma~\ref{lem:sign_vectors}).

Moreover, since $w^\circ$ and $w^*(\pi)$ have the same underlying permutation, there is a unique $\vec{\delta}^*\in\mathbb{F}_2^{2n}$ such that
\begin{equation} \label{eq-ap:lF2-system}
    w^*(\pi)= w^\circ  D(\vec{\delta}^*) \,.
\end{equation}
The previous equations implies that on every connected component $C$ of $H(S)$,
\begin{equation}
    \prod_{j\in C}(-1)^{\delta_j^*}=\left(\prod_{j\in C}s_j^*(\pi)\right)\left(\prod_{j\in C}s_j(w^\circ)\right)=1\,,
\end{equation}
where in the last equality we used that 
\begin{equation}
    \prod_{j\in C}s_j^*(\pi) = \operatorname{sign}(\pi_C)\,,\qquad\prod_{j\in C}s_j(w^\circ) = \operatorname{sign}(\pi_C)\,,
\end{equation}
 with $\pi_C$ the induced sub-permutation on $C$ (this follows because a product of Majorana braids always has determinant $+1$).
Thus, $\vec{\delta}^*$ has even Hamming weight on every component and can be realized by a choice of $\varepsilon_1,\ldots,\varepsilon_L$ in Eq.~\eqref{eq-ap:oriented-routing}. The orientations can be efficiently found classically by solving the associated linear system over $\mathbb{F}_2$, namely, Eq.~\eqref{eq-ap:lF2-system}.
\end{proof}

Lemma~\ref{lem:orientation} directly implies that given an instance $\pi$ of token swapping on a tree $T$, we can reduce it to solving the optimal synthesis problem for $w^*$, as an optimal sequence for $w^*$ gives an optimal sequence for $\pi$ by simply replacing the braids by swaps. For if that were not optimal for $\pi$, there would exist a shorter sequence, but also orientations $\varepsilon_1,\ldots,\varepsilon_L$ realizing $w^*$, contradicting the assumption that the original sequence was optimal for $w^*$ in the first place. Hence the decision version of gate-optimal matchgate-Clifford synthesis with Majorana braids on tree connectivity graphs is NP-hard.

To prove membership in NP, we need a polynomial upper bound on the length of a synthesis witness. Indeed, given a polynomial-length witness sequence, it can be efficiently verified whether it actually implements the target permutation matrix. And if the length threshold exceeds the bound, the instance is automatically a yes-instance, which concludes the proof.

\begin{lemma}
\label{lem:polynomial-bound}
For every tree $T$ on $[2n]$ and every $w\in B_{2n}\cap\SO(2n)$,
\begin{equation}
    \ell_T(w)\leq 3n(2n-1)\,.
\label{eq:polynomial-bound}
\end{equation}
\end{lemma}

\begin{proof}
Let $\pi=|w|$. Any permutation can be routed on a tree using at most $\binom{2n}{2}$ swaps. First, we move the token bound for a leaf to that leaf using at most $2n-1$ swaps, remove the leaf, and continue recursively. The next leaf will require at most $2n-2$ swaps and so on, so at most $(2n-1)+(2n-2)+\ldots 1=\binom{2n}{2}$ swaps suffice. We orient all braids in this routing in the same direction, and let $w^\circ$ be their product. Since $w$ and $w^\circ$ have the same underlying permutation and both have determinant $+1$, there exists an even-weight vector $\vec{\delta}\in\mathbb{F}_2^{2n}$ such that $w= w^\circ D(\vec{\delta})$.

We can implement $D(\vec{\delta})$ using at most $2n(2n-1)$ braids. To see this, let $j,j'$ be vertices of $T$ such that $\delta_j=\delta_{j'}=1$, and let $P= x_0,x_1,\ldots,x_d$, with $x_0=j$ and $x_d=j'$, be the unique path between them in $T$. Using the fact that $r_{x_{i-1}x_i}^2$ flips the signs at the two endpoints of the corresponding edge, ($x_{i-1},x_i)$, we find
\begin{equation}
    D(\vec{e}_j+\vec{e}_{j'})=\prod_{i=1}^{|P|} r_{x_{i-1}x_i}^2\,.
\end{equation}
Since $|P|\leq 2n-1$ on a tree with $2n$ vertices, implementing $D(\vec{e}_j+\vec{e}_{j'})$ requires at most $2(2n-1)$ generators. Moreover, there are at most $n$ pairs of ones in $\vec{\delta}$ (recall that $\vec{\delta}$ has even Hamming weight), and therefore
\begin{equation}
    \ell_T(w)\leq\binom{2n}{2}+2n(2n-1)= 3n(2n-1)\,.
\end{equation}

\end{proof}

\end{proof}

Next, we consider optimal synthesis with respect to depth.

\begin{theorem} 
    The decision version of depth-optimal matchgate-Clifford synthesis with Majorana braids on tree connectivity graphs is NP-complete.
\end{theorem}

\begin{proof}
    In this case, the reduction is from the version of token swapping which asks to minimize the number of parallel rounds of swaps, where every round contains pairwise disjoint swaps. This problem was proven NP-complete in Ref.~\cite{aichholzer2021hardness}. The reduction works exactly in the same way as for Theorem~\ref{thm:trees_app}, using $w^*$ to prove NP-hardness. To show inclusion in NP, we just note that the optimal depth cannot be larger than the optimal gate count $\ell_T(w)$ from Lemma~\ref{lem:polynomial-bound}. 
\end{proof}

\section{Proof of Corollary~\ref{cor:connected-graphs}}
\label{ap:general-graphs}

In this appendix, we provide a polynomial upper bound on the number of Majorana braids required to synthesize an arbitrary matchgate-Clifford unitary on a connected graph.

Let $G=(V,E)$ be a connected graph on the vertex set $V=[2n]$. For every edge $(a,b)\in E$, we include the generators $r_{ab}$ and $r_{ab}^{-1}$, and denote by $\ell_G(w)$ the minimum number of such generators required to synthesize $w\in B_{2n}\cap \SO(2n)$. We prove the following lemma.

\begin{lemma} \label{lem:connected-graphs}
For every connected graph $G$ on $[2n]$ and every
$w\in B_{2n}\cap \SO(2n)$,
\begin{equation}
    \ell_G(w) \leq 3n(2n-1)\,.
\end{equation}
\end{lemma}

\begin{proof}
Since $G$ is connected, it contains a spanning tree $T$ on the same vertex set (obtained for instance by successively removing edges from cycles while preserving connectedness).
Hence, every Majorana braid allowed on $T$ is also an allowed generator on $G$, which implies
\begin{equation}
    \ell_G(w) \leq \ell_T(w).
\end{equation}
By Lemma~\ref{lem:polynomial-bound}, we arrive at
\begin{equation}
     \ell_G(w) \leq \ell_T(w) \leq 3n(2n-1)\,.
\end{equation}
\end{proof}

Therefore, combining Theorems~\ref{th:tree-npcomplete} and~\ref{th:tree-depth} with Lemma~\ref{lem:connected-graphs} we obtain Corollary~\ref{cor:connected-graphs}, that we restate for convenience.
 \begin{corollary}
    The decision version of gate-optimal and depth-optimal matchgate-Clifford synthesis with Majorana braids on general connectivity graphs is NP-complete.
\end{corollary}

\end{document}